\documentclass{article}

\usepackage{arxiv}

\usepackage[utf8]{inputenc} % allow utf-8 input
\usepackage[T1]{fontenc}    % use 8-bit T1 fonts
\usepackage{hyperref}       % hyperlinks
\usepackage{url}            % simple URL typesetting
\usepackage{booktabs}       % professional-quality tables
\usepackage{amsfonts}       % blackboard math symbols
\usepackage{nicefrac}       % compact symbols for 1/2, etc.
\usepackage{microtype}      % microtypography
\usepackage{lipsum}		% Can be removed after putting your text content
\usepackage{graphicx}
\usepackage[numbers,square]{natbib}
\usepackage{doi}

\usepackage{amsmath,amsthm,amssymb}
\usepackage{newtxmath}

\usepackage{mdframed}

\usepackage{amsmath,amsthm,amssymb}
\usepackage{pgfplots}
\pgfplotsset{compat=1.18}
\newtheorem{definition}{Definition}
\usepackage{threeparttable}

\usepackage{booktabs}   % For professional table lines
\usepackage{subcaption} % The essential package for sub-tables/figures
\usepackage{newfloat}
\usepackage{bm}
\usepackage{listings}
\usepackage{booktabs} % For professional looking tables
\usepackage{siunitx}  % For decimal alignment
\usepackage{balance}
\usepackage{cancel}
\usepackage{xcolor}      % For using colors
\definecolor{darkgreen}{rgb}{0, 0.39, 0}  % RGB 值接近 HTML 中的 #006400

\usepackage{graphicx}    % For \rotatebox
\usepackage{multirow}    % For multi-row cells if needed
\usepackage{threeparttable}
\usepackage{xspace}

\DeclareMathOperator*{\argmax}{arg\,max}

\usepackage{colortbl}


\usepackage{tikz}
\usetikzlibrary{shapes, positioning, arrows.meta, calc, fit, backgrounds}

\definecolor{codegreen}{rgb}{0,0.6,0}
\definecolor{codegray}{rgb}{0.5,0.5,0.5}
\definecolor{codepurple}{rgb}{0.58,0,0.82}
\definecolor{backcolour}{rgb}{0.98,0.98,0.98}

\definecolor{randomcolor}{RGB}{255,107,107}
\definecolor{ctxpipecolor}{RGB}{255,165,0}
\definecolor{optimalcolor}{RGB}{81,207,102}
\definecolor{highlightcolor}{RGB}{102,126,234}

\usepackage{enumitem}
\def\BibTeX{{\rm B\kern-.05em{\sc i\kern-.025em b}\kern-.08em
    T\kern-.1667em\lower.7ex\hbox{E}\kern-.125emX}}

\usepackage{mathtools}
\usepackage[ruled,linesnumbered, vlined]{algorithm2e}
\usepackage[normalem]{ulem}

\SetKwInput{KwParam}{Param}
\SetKwProg{Procedure}{procedure}{}{}
\newcommand{\kw}[1]{{\ensuremath {\mathsf{#1}}}\xspace}

\newtheorem{theorem}{Theorem}

\newtheorem{example}{Example}
\usepackage{mdframed}

\newcommand{\learner}{\mathcal{A}}
\newcommand{\trainset}{D_{train}}

\newcommand{\rawset}{D_\kw{raw}}
\newcommand{\pipeline}{P}
\newcommand{\optimalpipeline}{P^{\star}}
\newcommand{\library}{\mathcal{O}}
\newcommand{\libraryplus}{\mathcal{O}_{+}}
\newcommand{\categories}{\mathcal{C}}
\newcommand{\orders}{\prec\!\prec_\library}
\newcommand{\ordersplus}{\prec\!\prec_{\libraryplus}}
\newcommand{\maxlen}{\widebar{M}}
\newcommand{\sv}{\psi^{SV}}
\newcommand{\psv}{\phi^{PSV}}
\newcommand{\svcalc}{\mathcal{F}}
\newcommand{\noofsamples}{n_{\kw{sample}}}
\newcommand{\catepipe}{\kw{CP}}
\newcommand{\mabscore}{x}

\newcommand{\E}{\mathbb{E}}

\usepackage[normalem]{ulem}

\newcommand{\myparagraph}[1]{%
    \noindent{\bfseries\uline{#1}}
}

\usepackage{hyperref}

\newcommand{\suppurl}{{\url{https://github.com/shappipe/shapleypipe}}}
\newcommand{\supp}[1]{Supp.~#1}

\newcommand{\fixed}[1]{{\color{black}#1}}

\usepackage{cleveref}

\usepackage{marvosym}

\usepackage{marvosym}

\title{ShapleyPipe: A Game-Theoretic Approach to Data Preparation Pipeline Construction}

\author{%
	Jing Chang \\
	Shenzhen University \\
	\texttt{2350273007@email.szu.edu.cn}
	\And
	Shuyuan Zheng\textsuperscript{\Letter} \\
	Osaka University \\
	\texttt{zheng@ist.osaka-u.ac.jp}
	\And
	Jinbin Huang \\
	Shenzhen University \\
	\texttt{jbhuang@szu.edu.cn}
	\And
	Chang Liu \\
	Shenzhen University \\
	\texttt{2300271084@email.szu.edu.cn}
	\And
	Jianbin Qin\textsuperscript{\Letter} \\
	Shenzhen University \\
	\texttt{qinjianbin@szu.edu.cn}
}

\renewcommand{\shorttitle}{ShapleyPipe}

\begin{document}
\maketitle

\begin{abstract}
Data preparation is a critical yet highly time-consuming component of machine learning (ML) workflows and has a substantial impact on downstream model performance.
Automated data preparation pipeline construction requires selecting appropriate operators, such as missing-value imputation, feature scaling, and feature transformation, and placing them at suitable pipeline positions. However, existing methods lack a principled and interpretable mechanism for understanding the marginal contribution of each operator at a
specific position. To address this limitation, we study data preparation pipeline construction from a cooperative game-theoretic perspective. We formulate the task as a pipeline construction game, in which operators act as players that cooperatively improve downstream model utility. Based on this formulation, we propose the Pipeline Shapley value (PipeShapley), a contribution-attribution method that explicitly
captures position-dependent contributions and distinguishes repeated occurrences of the same operator type.
Building on PipeShapley, we introduce ShapleyPipe, an interpretable pipeline construction framework that selects and orders operators according to their position-aware contributions. To make PipeShapley practical, we develop Hierarchical ShapleyPipe (H-SP) based on a Category Coherence Assumption to exploit coherent contribution patterns, combining MAB-based representative selection with adaptive early termination. Experiments on 27 classification datasets show that H-SP achieves the highest average accuracy of 0.847 and the best average rank among the evaluated practical methods, while lying on the empirical
accuracy--runtime Pareto frontier. On three datasets with up to 581K records, 10K-subsample search and 16-worker parallel execution reduce the average wall-clock time to 891 seconds while retaining the highest average accuracy of 0.705.
\end{abstract}

% keywords can be removed

\makeatletter
\newlength{\mytablewidth}
\setlength{\mytablewidth}{\columnwidth}
\makeatother

% \section{Introduction}
% \lipsum[2]
% \lipsum[3]

\begin{figure*}[t]
    \centering
    \begin{subfigure}[]{0.32\linewidth}
        \includegraphics[width=\linewidth]{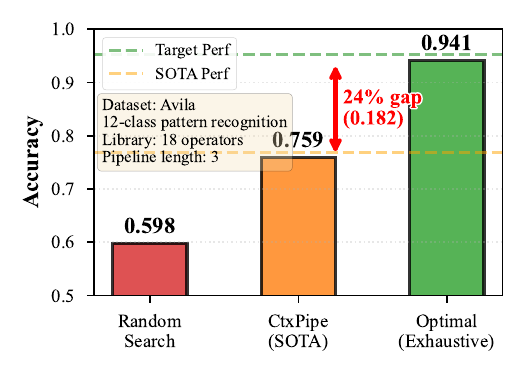}
        \caption{Performance Gap on Avila}
    \end{subfigure}
    \hfill
    \begin{subfigure}[]{0.32\linewidth}
        \includegraphics[width=\linewidth]{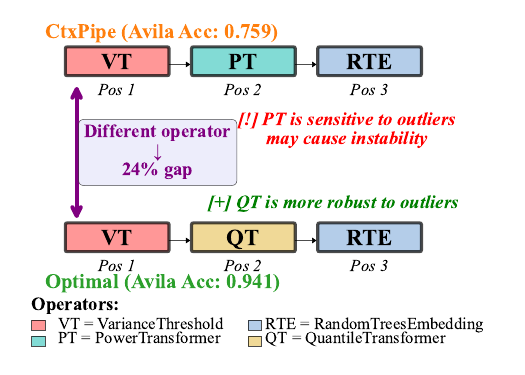}
        \caption{Changing Operator Selection}
    \end{subfigure}
    \hfill
    \begin{subfigure}[]{0.32\linewidth}
        \includegraphics[width=\linewidth]{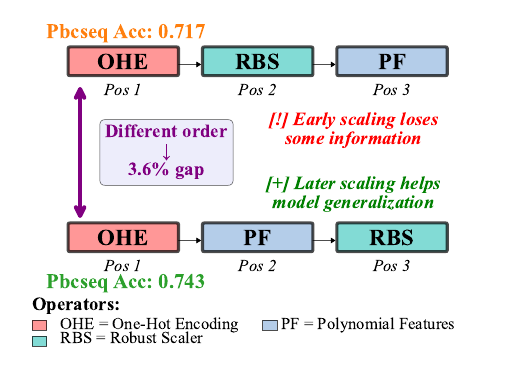}
        \caption{Changing Operator Order}
    \end{subfigure}
    \caption{Operator selection and ordering both matter. (a) On the \fixed{Avila dataset (20{,}867 samples, 11 features, 12 classes; a pattern recognition task)}, CtxPipe vs. Optimal shows a 24\% performance gap. (b) Operator selection matters on Avila: Replacing PT with QT results in a $24\%$ performance gap. (c) On the \fixed{Pbcseq dataset (1{,}945 samples, 19 features, 2 classes; a clinical survival analysis task)}, operator ordering matters: exchanging RBS and PF results in a $3.6\%$ performance gap.}
    \label{fig:motivating_example}
\end{figure*}

\section{Introduction}
\label{sec:introduction}

Data preparation remains the most time-consuming phase in machine learning workflows, often consuming a substantial fraction of the overall effort~\cite{stonebraker2018data,kandel2012enterprise}. Effective data preparation pipelines consist of a sequence of operations, such as missing value imputation, feature scaling, and feature selection, and their construction critically determines downstream model performance~\cite{garcia2016data}. However, this construction process relies heavily on manual expertise and repeated trial-and-error, making it a prime target for automation. Although Automated Machine Learning (AutoML) has achieved significant progress in automating algorithm selection and hyperparameter tuning~\cite{feurer2015efficient,olson2016evaluation}, the construction of data preparation pipelines still faces a fundamental challenge: combinatorial explosion. Given $N$ distinct operators for at most $\maxlen$ pipeline positions, the search space contains $O(N^{\maxlen})$ possible configurations. This renders exhaustive pipeline construction infeasible and necessitates efficient methods to select optimal operators and place them at appropriate positions.

% Existing automated pipeline construction methods can be broadly categorized into three classes, each with critical limitations.
% Search-based methods, including grid search, random search~\cite{bergstra2012random}, and Bayesian optimization~\cite{thornton2013auto,shahriari2015taking,snoek2012practical}, struggle with high-dimensional discrete search spaces and typically require thousands of evaluations.
% Evolutionary algorithms, such as TPOT~\cite{olson2016evaluation} and SAGA~\cite{siddiqi2023saga}, are prone to premature convergence and lack theoretical guarantees~\cite{eiben2015evolutionary}.
% Reinforcement learning (RL) approaches, including CtxPipe~\cite{gao2024ctxpipe} and HAI-AI~\cite{chen2023haipipe}, represent the current state of the art and rely on sparse, end-to-end rewards to guide sequential operator selection.
% However, because the reward signal is shared across all pipeline positions, the learned policies cannot distinguish the contributions of operators at different positions, which may lead to pipelines with suboptimal operator ordering, as illustrated in the following example.

Existing automated pipeline construction methods fall into four broad classes.
\emph{Search-based} methods, including grid search, random
search~\cite{bergstra2012random}, and Bayesian
optimization~\cite{thornton2013auto,shahriari2015taking,snoek2012practical}, treat a pipeline as a black-box configuration. In high-dimensional discrete spaces they typically need thousands of evaluations. \emph{Evolutionary} algorithms, such as TPOT~\cite{olson2016evaluation} and SAGA~\cite{siddiqi2023saga}, are prone to premature convergence and offer few
guarantees~\cite{eiben2015evolutionary}.
\emph{Reinforcement learning (RL)} approaches, including
Learn2Clean~\cite{berti2019learn2clean}, CtxPipe~\cite{gao2024ctxpipe} and
HAIPipe~\cite{chen2023haipipe}, learn sequential operator-selection policies from sparse, end-to-end rewards. \emph{LLM-based} approaches, such as
CatDB~\cite{catdb} and AutoML-Agent~\cite{automlagent}, prompt a large language model with dataset metadata to propose pipelines or code directly. The RL- and LLM-based classes share a limitation: the feedback they receive, an end-to-end reward or a final validation score, is attached to the pipeline as a whole. Neither can tell how much an operator contributes \emph{at a given position}, which can lead to pipelines with suboptimal operator selection and
ordering, as the following example illustrates.

% \textbf{Motivating Examples.} To illustrate the fundamental challenge in automated pipeline construction, we designed a controlled experiment on the Avila dataset. We constructed a search space with 18 operators and pipeline length 3, yielding $18^3 = 5{,}832$ possible configurations, small enough for complete exhaustive evaluation but large enough to be meaningful. Pipeline performance is measured by training a LogisticRegression classifier on the pipeline-transformed data and evaluating its prediction accuracy on a validation set.

% We reproduced CtxPipe~\cite{gao2024ctxpipe} within this search space and compared it against exhaustive search. As shown in Figure~\ref{fig:motivating_example}a, CtxPipe achieves 0.759 accuracy, while exhaustive search reveals the true optimal pipeline achieving 0.941 accuracy. There is a substantial 24\% performance gap between them. This experiment demonstrates that existing methods fail to discover high-quality pipelines even when the optimal solution exists within their search space.
% Analyzing the two pipelines (Figure~\ref{fig:motivating_example}b), we observe that they differ primarily in operator selection: CtxPipe selects PowerTransformer at position 2, which is sensitive to outliers and may cause instability, while the optimal pipeline uses the more robust QuantileTransformer. This may suggest that operator selection impacts performance, but is this always the dominant factor?

\textbf{Motivating Examples.} To make the challenge concrete, we built a controlled search space on the Avila dataset with 18 operators and pipeline length 3, i.e., $18^3=5{,}832$ configurations, few enough to enumerate exhaustively. Each pipeline is scored by training a LogisticRegression classifier on the transformed training split
and measuring accuracy on a held-out validation split, averaged over 5 seeds. We ran CtxPipe~\cite{gao2024ctxpipe} and random search with 20 evaluations, and compared them with exhaustive enumeration (Figure~\ref{fig:motivating_example}a). CtxPipe reaches 0.759 accuracy, whereas the best enumerated pipeline reaches 0.941, an absolute gap of 0.182 (24\% relative). The optimal pipeline lies inside CtxPipe's search space, yet CtxPipe does not find it. The two pipelines differ in a single operator at a single position (Figure~\ref{fig:motivating_example}b): CtxPipe places PowerTransformer at position~2, whereas the optimal pipeline uses QuantileTransformer there. QuantileTransformer is less sensitive to outliers, which are frequent in Avila (4{,}458 of its 20{,}867 records; \supp{C.1}\footnote{``Supp.~X'' refers to Section~X of the supplementary material (additional results, code, and data), available at \suppurl.}). This suggests that operator selection matters, but is it always the dominant
factor?

To validate this hypothesis on an even smaller and more controllable scale, we designed a second experiment on the Pbcseq dataset, exhaustively evaluating all $5^3 = 125$ pipelines from 5 selected operators. Figure~\ref{fig:motivating_example}c reveals a different insight: two pipelines using the \textit{exact same three operators}: OneHotEncoder, PolynomialFeatures, and RobustScaler but in different orderings, achieve 0.743 vs.\ 0.717 accuracy (a 3.6\% gap). The suboptimal ordering applies RobustScaler before PolynomialFeatures; the optimal ordering delays scaling until after feature engineering. This demonstrates that operator ordering can be as critical as operator selection.

The above experiments reveal two key insights: (1) state-of-the-art methods leave substantial performance on the table, (2) both operator selection and ordering matter.
However, none of the existing methods can systematically explain the marginal contribution that each operator makes at a specific position in the pipeline, making principled operator selection and ordering difficult.
We argue that the second insight is likely a key reason for the first, therefore this paper proposes an interpretable, \textit{position-aware contribution attribution} method for operators to guide data preparation pipeline construction (DPPC).

Inspired by the data selection task, we observe that cooperative game theory provides a promising contribution attribution tool for guiding operator selection and ordering.
Specifically, in the seminal work of data selection~\cite{ghorbani2019data}, Ghorbani et al.\ treat each data sample in a training dataset as a player, and view the machine learning task as players in the dataset cooperating to train an ML model. The Shapley value (SV)~\cite{shapley1953value} is then used as a contribution attribution metric to compute the marginal contribution of each data sample, and data samples are selected based on their SVs.
Similarly, in the DPPC task, we observe that each operator can be regarded as a player, and operator selection and ordering can be guided by its marginal contribution to model utility.

However, there are substantial differences between data selection and operator selection/ordering, which pose two fundamental challenges to applying cooperative game theory for explaining and guiding DPPC.
First, in data selection, due to data shuffling during the ML process, the initial position of a data sample in the training set has only a negligible effect on its marginal contribution to model utility.
In contrast, in DPPC, as discussed above, the position of an operator within the pipeline has a significant impact on model utility.
Second, in data selection, each data sample can be selected at most once, whereas in a pipeline, the same operator may appear at multiple positions.
These differences make Shapley-value-based data selection tools ill-suited to the requirements of pipeline construction.
Therefore, new position- and repetition-aware contribution attribution methods are demanded for DPPC.

To the best of our knowledge, this paper is the first to study the DPPC problem from a game-theoretic perspective.
By extending traditional cooperative game theory to address the unique challenges of DPPC, we propose a novel approach named ShapleyPipe.
Compared with existing methods, ShapleyPipe improves model utility while enabling efficient, principled, and interpretable operator selection and ordering.
% Compared with existing methods, ShapleyPipe not only achieves significant improvements in model utility while maintaining computational efficiency, but also makes operator selection and ordering more interpretable and principled.
Our main contributions are summarized as follows:
\begin{itemize}[leftmargin=*]
    \item First, we formulate the DPPC task as a pipeline construction game and, based on this formulation, introduce a new contribution attribution tool named the Pipeline Shapley value (PipeShapley).
    Unlike the classical Shapley value in traditional cooperative games, PipeShapley facilitates position- and repetition-aware contribution attribution for operator selection and ordering.
    % \item Second, we propose ShapleyPipe, which constructs pipelines based on operators' PipeShapley values.
    % We also develop tailored approximation algorithms to accelerate the computation of PipeShapley values for ShapleyPipe.
    \item Second, we introduce and empirically validate the Category Coherence Assumption, which captures coherent contribution patterns among functionally related operators. Building on this assumption, we develop Hierarchical ShapleyPipe (H-SP), a two-stage approximation algorithm to accelerate the computation of PipeShapley values for ShapleyPipe.
    % with an approximation-error bound that substantially reduces pipeline utility evaluations.
    % \item Finally, we conduct extensive experiments on 27 datasets to evaluate the efficiency and effectiveness of ShapleyPipe. The results demonstrate that ShapleyPipe achieves substantially higher model utility than SOTA methods, while maintaining practical runtime.
    % \item Third, we identify when position-aware attribution pays off. We introduce positional sensitivity, a sampling-based probe that measures how strongly a dataset's utility depends on operator ordering, and show it predicts H-SP's advantage.
    % \item Finally, we conduct extensive experiments on 27 datasets and three large-scale datasets. H-SP attains the best accuracy–runtime trade-off among practical methods on the benchmark suite, while on large-scale data we show that subsampled search with parallel execution keeps it competitive.
    \item Finally, we evaluate on 27 benchmark datasets and three large datasets. H-SP attains the highest mean accuracy, and its advantage over every practical baseline is statistically significant ($p<0.01$). It also lies on the accuracy--runtime Pareto frontier. On large datasets, subsampled search with parallel execution keeps it competitive.

\end{itemize}

% \section{PROBLEM FORMULATION}\label{sec:prelimnaries}
\section{PROBLEM FORMULATION}\label{sec:prelimnaries}

In this section, we formally define the \textit{data preparation pipeline construction (DPPC) problem}.

Consider a machine learning (ML) task $(\learner, \trainset)$ where ML engineers apply a learning algorithm $\learner$ to training data $\trainset$ to yield an ML model $\learner(\trainset)$.
To optimize the model's predictive performance, the learning algorithm is designed to maximize the utility $u(\learner(\trainset))$ of the model based on some utility metric $u$, such as model accuracy and (negated) validation loss.
In addition to the learning algorithm, the quality of the training data also largely determines the utility of the ML model.
However, the raw data available to ML engineers is often not well organized or cleaned for training accurate ML models.
Therefore, before performing the ML task, we need to perform data preparation on the given raw dataset $\rawset$ to obtain a high-quality training set $\trainset$.

\myparagraph{Data preparation operator and pipeline.}
A \textit{data preparation operator} is a deterministic function $o: \mathcal{D} \to \mathcal{D}$ that transforms a dataset $D \in \mathcal{D}$ into a modified dataset $D'$~\cite{pyle1999data}, such as an encoder, a scaler, and a feature transformer. 
From a \textit{data operator library} $\library = \{o_1, \dots, o_N\}$ consisting of $N$ distinct data preparation operators, we can select multiple operators to construct a \textit{data preparation pipeline} $\pipeline = [p_1, \dots, p_M]$, where $M$ denotes the pipeline length and each $p_j \in \library$ for all data preparation steps $j \in [M]$.
Note that the pipeline may include repeated data preparation operators.
After constructing the pipeline, it can be sequentially applied to the raw dataset to obtain the training data, i.e., $$\trainset = \pipeline(\rawset) \coloneqq p_M \circ p_{M-1} \circ \dots \circ p_1 (\rawset).$$

\myparagraph{Partial order.}
Moreover, the order in which certain operators appear in the pipeline should satisfy some given constraints, which are defined by partial order relations.
Formally, a partial order $o_i \prec o_j$ requires that operator $o_i$ should appear in the pipeline before operator $o_j$, and we denote by $\orders$ the set of partial orders imposed on the operator library $\library$. 

\myparagraph{The DPPC problem.}
Finally, we aim to solve the DPPC problem to identify an optimal data preparation pipeline that maximizes the utility $u(\learner(P(\rawset)))$ of the resulting ML model $\learner(P(\rawset))$.
Since in the DPPC problem the learning algorithm $\learner$ and the original dataset $\rawset$ are given, the model utility is solely determined by the pipeline $\pipeline$. 
To reduce notational overload, we reuse the notation to denote the utility of the model $\learner(P(\rawset))$ as $u(P)$.
Consequently, given a raw dataset $\rawset$, a learning algorithm $\learner$, a data operator library $\library$, and a set of partial orders $\orders$, the DPPC problem is defined as:
\begin{align}
\label{obj:unbounded_dppc}
    \optimalpipeline = \argmax_{P,\text{ s.t. } \orders} u(P).
\end{align}
% $$\optimalpipeline = \argmax_{P,\text{ s.t. } \orders} u(\learner(P(\rawset))).$$

% \input{preliminaries}

\section{Proposed Method}
\label{sec:pipeshapley}
In this section, we analyze the DPPC problem from a cooperative game-theoretic perspective and propose a Shapley value-based approach to address it.

\subsection{Preliminaries}
\label{sec:Preliminaries}
We begin by introducing some necessary background information in game theory.

\myparagraph{Cooperative game.}
Cooperative game theory addresses how to allocate a payoff among multiple players who collaborate to complete a task, in a manner that fairly reflects their individual contributions to the task. 
Formally, a cooperative game $(\mathcal{N}, v)$ is specified by a finite set of players $\mathcal{N}=\{1,\dots,N\}$ and a payoff function $v: 2^{\mathcal{N}} \rightarrow R$ that maps any possible coalition $S \subseteq \mathcal{N}$ to a payoff $v(S)$.
A solution concept $\psi$ to the cooperative game $(\mathcal{N}, v)$ returns a vector $\psi(\mathcal{N}, v)$ where $\psi_i(\mathcal{N}, v)$ represents the contribution of player $i\in \mathcal{N}$ to accomplishing the task $v(\mathcal{N})$.

\myparagraph{Shapley value (SV).}
The SV is a well-known and widely-adopted solution concept in cooperative game theory, which attributes the average payoff improvement by each player as their contribution.
It is considered as a promising solution concept because it uniquely satisfies some fundamental properties for fair contribution attribution, including efficiency, symmetry, linearity, and the dummy player.
Specifically, for each player $i$, the SV $\sv_i$ enumerates all possible coalitions $S \subseteq \mathcal{N}$, calculates the payoff improvement $v(S \cup \{i\}) - v(S)$ for each $S$, and averages all the improvements across the coalitions: 
\begin{align}
    \sv_i(\mathcal{N}, v) = \sum_{S \subseteq \mathcal{N} \setminus \{i\}} \frac{|S|! (|\mathcal{N}| - |S| - 1)!}{|\mathcal{N}|!} [v(S \cup \{i\}) - v(S)].
\end{align}

\begin{example}[SV Calculation for 3-Player Game]
    Consider players $\{A, B, C\}$ and a payoff function defined as: $v(\emptyset) = v(\{A\}) = v(\{B\}) = v(\{C\}) = 0$, $v(\{A, B\}) = 50$, $v(\{A, C\}) = 30$, $v(\{B, C\}) = 40$, and $v(\{A, B, C\}) = 80$. The payoff improvements by player $A$ are calculated as follows:
\begin{itemize}
    \item Joining the empty coalition $S = \emptyset$: $v(\{A\}) - v(\emptyset) = 0$.
    \item Joining $S = \{B\}$: $v(\{A, B\}) - v(\{B\}) = 50$.
    \item Joining $S = \{C\}$: $v(\{A, C\}) - v(\{C\}) = 30$.
    \item Joining $S = \{B, C\}$: $v(\{A, B, C\}) - v(\{B, C\}) = 40$.
\end{itemize}
Then, player $A$'s SV is calculated as $\sv_A = \frac{2}{6}\cdot 0 + \frac{1}{6} \cdot 50 + \frac{1}{6} \cdot 30 + \frac{2}{6}\cdot 40 = \frac{80}{3}$.
\end{example}

\myparagraph{Challenges in solving DPPC as a cooperative game.}
In the DPPC problem, constructing a data preparation pipeline requires selecting operators that provide greater improvements to the utility of the ML model.
Intuitively, these operators can be viewed as players in a cooperative game, collaborating to complete the data preparation task and receiving the utility of the ML model as the payoff. 
However, we identify two key challenges in addressing the DPPC problem through constructing a cooperative game.
\begin{itemize}[leftmargin=*]
    \item C1 (Position-Aware Contribution Attribution): 
    The focus of a cooperative game is solely on attributing contributions among players. 
    In contrast, the DPPC problem requires consideration of each operator's performance at different positions within the pipeline. 
    This prevents a solution concept from cooperative game theory from conveying a meaningful signal regarding which operator should be selected for each data preparation step $j$.
    % In cooperative game theory, the payoff function does not account for the order in which players participate in the game; its domain is $2^{\mathcal{N}}$. 
    % In contrast, the utility function in the DPPC problem must consider the sequential order in which operators process the data.
    \item C2 (Repetition-Aware Contribution Attribution): Cooperative game theory assumes that each player participates only once. 
    However, the DPPC problem allows each operator to appear at multiple positions within the pipeline. 
    Therefore, solution concepts from cooperative game theory cannot inform us how many times a given operator should appear in the pipeline.
\end{itemize}
Therefore, although addressing the DPPC problem through cooperative game theory sounds promising, theoretical and technical advances are still required to tackle the unique challenges posed by the DPPC setting.

\subsection{Pipeline Construction Game}
To address the aforementioned challenges, we extend the traditional cooperative game theory and propose the concept of \textit{pipeline construction game}.
In this game, in addition to assigning a contribution $\phi_{i}$ to each operator $o_i \in \library$, we aim to further distribute $\phi_{i}$ across all positions in the pipeline to obtain a position-aware contribution $\phi_{i,j}$ for each position $j$.
However, in \fixed{Equation~\eqref{obj:unbounded_dppc}}, the pipeline length is uncertain and could be arbitrarily long before the optimal length is determined, which makes position-aware attribution infeasible. 

Therefore, we introduce an upper bound $\maxlen$ on the pipeline length $M$, enforcing the constraint $M \leq \maxlen$ for the DPPC problem. 
\fixed{This assumption is commonly made in existing DPPC methods~\cite{gao2024ctxpipe,olson2016evaluation,autosklearn,chen2023haipipe}, yielding a practically bounded search space.}
% This assumption is common in existing DPPC methods, where it is used to reduce the pipeline search space.
% \fixed{This assumption is widely adopted in existing DPPC methods to
% reduce the pipeline search space. Specifically,
% CtxPipe~\cite{gao2024ctxpipe} sets a maximum pipeline length of
% $\maxlen=6$; TPOT~\cite{olson2016evaluation} provides a
% \textit{max\_pipeline\_length} parameter to bound the search depth;
% Auto-sklearn~\cite{autosklearn} constrains pipeline depth through
% its meta-learning configuration space; and HAI-AI~\cite{chen2023haipipe}
% similarly limits the number of sequential operations.}
% Beyond computational motivation, our sensitivity analysis in Section~\ref{subsec:sensitivity} empirically validates this assumption
% by showing that pipelines beyond moderate lengths tend to suffer from over-transformation.
Additionally, let us introduce a mock operator $o_0$ that performs nothing over the data, which results in a new operator library $\libraryplus = \library \cup \{o_0\}$ of size $N+1$.
Under the bounded pipeline length assumption, the DPPC problem is reduced to finding an ordered sequence of $\maxlen$ (possibly repeated) operators from the new library $\libraryplus$ that maximizes model utility, where $M$ positions are occupied by real operators and $o_0$ serves as a placeholder for the remaining $\maxlen - M$ positions.
Therefore, the objective of the DPPC problem becomes:
\begin{align}
\label{obj:bounded_dppc}
    \optimalpipeline = \argmax_{P\in \libraryplus^{\maxlen},\text{ s.t. } \ordersplus} u(P),
\end{align}
% $$\optimalpipeline = \argmax_{P\in \libraryplus^{\maxlen},\text{ s.t. } \orders} v(\learner(P(\rawset))),$$
where $\libraryplus^{\maxlen}$ is the Cartesian product of the set $\libraryplus$ taken $\maxlen$ times.

\begin{definition}[Pipeline Construction Game]
\label{def:pcgame}
    A pipeline construction game $(\libraryplus, \mathcal{M}, u)$ consists of a finite set of players $\libraryplus$, a set of positions $\mathcal{M} = \{1,\dots,\maxlen\}$
    and a utility function $u: \libraryplus^{\maxlen} \rightarrow R$. 
    A solution concept $\phi$ to the game $(\libraryplus, \mathcal{M}, u)$ returns a matrix $\phi(\libraryplus, \mathcal{M}, u)$ where each element $\phi_{i, j}(\libraryplus, \mathcal{M}, u)$ represents the contribution of operator $o_i$ to accomplishing the task $\mathbb{E}_{P\in \libraryplus^{\maxlen},\text{ s.t. } \ordersplus}[u(P)]$ when $o_i$ is in the $j$-th position of the data preparation pipeline. 
\end{definition}

% \fixed{the bounded-length DPPC formulation (Equation~\eqref{obj:bounded_dppc})}

Based on \fixed{the objective in Equation~\eqref{obj:bounded_dppc}}, we define the pipeline construction game $(\libraryplus, \mathcal{M}, u)$ in Definition~\ref{def:pcgame}.
When no ambiguity arises, we abbreviate the solution concept $\phi_{i, j}(\libraryplus, \mathcal{M}, u)$ of the game as $\phi_{i,j}$.

\subsection{Pipeline Shapley Value and ShapleyPipe}
Then, for the pipeline construction game, we propose a solution concept named PipeShapley (Pipeline Shapley value).
The design of PipeShapley is based on the SV. 
In the following, we explain how PipeShapley is formulated.

First, the SV can also be expressed in the following permutation-based form:
\begin{align*}
    \sv_i(\mathcal{N}, v) = \frac{1}{N!}\sum_{\pi \in \Pi(\mathcal{N})} [v(S_{i}^{\pi} \cup \{i\}) - v(S_{i}^{\pi})],
\end{align*}
where $\Pi(\mathcal{N})$ is the set of all possible permutations of $\mathcal{N}$ and $S_{i}^{\pi}$ is the set of players in $\mathcal{N}$ who succeeds $i$ in the permutation $\pi$. 
Similarly, assuming that each operator $o_i$ appears at most once in the pipeline, we can construct the following function $\phi_{i}^{v1}$ to represent the contribution of $o_i$:
\begin{align*}
    \phi_{i}^{v1}(\libraryplus, \mathcal{M}, u) = \mathbb{E}_{P \in \Pi(\libraryplus)} [u(P_{\geq o_i}) - u(P_{> o_i})],
\end{align*}
where $P_{> o_i}$ denotes the subsequence of pipeline $P$ following $o_i$, and $P_{\geq o_i}$ is obtained by inserting $o_i$ at the beginning of $P_{> o_i}$.
\fixed{Note that the payoff function $v$ is different from the utility function $u$: the former takes as input an unordered set of players, whereas the latter takes as input an ordered sequence of players, namely a pipeline.}

% \fixed{
% \smallskip\noindent\textbf{From Position-Agnostic to Position-Aware.}\; The formulation $\phi^{v1}_i$ aggregates operator $o_i$'s marginal contributions across \emph{all} positions where $o_i$ happens to appear in a random permutation, producing a single scalar per operator. While this mirrors the classical Shapley value, it loses positional information: two operators may have identical $\phi^{v1}$ values yet contribute very differently at specific pipeline positions. To address challenge C1, we decompose $\phi^{v1}_i$ into position-specific components by restricting the marginal contribution to only those permutations where $o_i$
% occupies a particular position $j$, yielding the position-aware formulation $\phi^{v2}_{i,j}$:
% }
% We then further extend $\phi_{i}^{v1}$ to make it position-aware, yielding the following solution concept:
% \begin{align*}
%     \phi_{i, j}^{v2}(\libraryplus, \mathcal{M}, u) = \mathbb{E}_{P \in \Pi(\libraryplus)} \big[\boldsymbol{1}[p_j = o_i] \cdot (u(P_{\geq j}) - u(P_{> j}))\big],
% \end{align*}
% where $P_{\geq j}$ ($P_{> j}$, resp.) denotes the subsequence of pipeline $P$ starting from position $j$ (position $j+1$, resp.). 
% Intuitively, $\phi_{i, j}^{v2}$ considers only the utility improvements by $o_i$ when it occupies position $j$ in the pipeline, treating these position-aware improvements as its contribution in position $j$.

\fixed{We then further extend $\phi_{i}^{v1}$ to make it position-aware. The crux of the extension is that $\phi_{i}^{v1}$ collapses an operator's contributions across all positions into a single scalar, whereas the DPPC problem requires us to know which specific position an operator should occupy. To recover this information, we make two coordinated modifications to $\phi_{i}^{v1}$, both centered on a target position $j$. 
(1) We replace the global marginal contribution $u(P_{\geq o_i}) - u(P_{> o_i})$ with a position-conditioned marginal contribution $u(P_{\geq j}) - u(P_{> j})$, where $P_{\geq j}$ ($P_{> j}$, resp.) denotes the subsequence of pipeline $P$ starting from position $j$ (position $j+1$, resp.). 
Intuitively, this measures the utility improvement induced by whichever operator occupies position $j$, rather than by $o_i$ at an unspecified location.
(2) Then, to attribute this improvement specifically to $o_i$ at position $j$, we weight the marginal contribution by the indicator $\boldsymbol{1}[p_j = o_i]$, which equals one only when the operator at position $j$ in the sampled permutation is indeed $o_i$.
Combining the two modifications yields a position-aware solution concept:}
\begin{align*}
    &\phi_{i, j}^{v2}(\libraryplus, \mathcal{M}, u) \!=\! \mathbb{E}_{P \in \Pi(\libraryplus)} \big[\boldsymbol{1}[p_j \!=\! o_i] \!\cdot\! (u(P_{\geq o_i}) \!-\! u(P_{> o_i}))\big] \\
    = & \mathbb{E}_{P \in \Pi(\libraryplus)} \big[\boldsymbol{1}[p_j = o_i] \cdot (u(P_{\geq j}) - u(P_{> j}))\big].
\end{align*}
\fixed{In effect, $\phi_{i, j}^{v2}$ averages the marginal contribution at position $j$ only over those permutations in which $o_i$ actually appears at that position, so the resulting value reflects the contribution attributable to placing $o_i$ specifically at position $j$. This per-position decomposition is what allows PipeShapley to inform the operator selection at each step of the pipeline, directly addressing challenge C1.}

\fixed{Finally, we adjust the pipeline search space to enable the same operators to appear at multiple positions, which aligns with the objective in Equation~\eqref{obj:bounded_dppc}, resulting in the formulation of PipeShapley:}
\begin{equation*}
\begin{aligned}
    \psv_{i,j}(\libraryplus, \mathcal{M}, u) \!=\! \E_{P\in \libraryplus^{\maxlen}
    | \ordersplus} \! \big [\boldsymbol{1}[p_j \! = \! o_i] \! \cdot \! (u(P_{\geq j}) \!-\! u(P_{> j})) \big].
\end{aligned}
\label{eq:pipeshapley}
\end{equation*}

% \fixed{\begin{example}[Position-wise attribution]\label{ex:pos}
% Consider two operators $\{a,b\}$ and positions $j\in\{1,2\}$. Because utility
% depends on order, we specify $u$ over sequences:
% $u([\,])=0,\ u([a])=2,\ u([b])=1,\ u([a,b])=5,\ u([b,a])=3$. For operator $a$:
% position~1 (only $[a,b]$ has $p_1=a$) gives marginal $5-1=4$, so
% $\phi^{v_2}_{a,1}=\tfrac12\!\cdot\!4=2$; position~2 (only $[b,a]$) gives marginal
% $2-0=2$, so $\phi^{v_2}_{a,2}=\tfrac12\!\cdot\!2=1$. The operator-level value is
% $\phi^{v_1}_a=3$, and indeed $\phi^{v_2}_{a,1}+\phi^{v_2}_{a,2}=2+1=3$,
% consistent with \Cref{rmk:decomp}. While $\phi^{v_1}_a$ reports only $3$,
% $\phi^{v_2}$ exposes that placing $a$ at position~1 is worth $2$ but at
% position~2 only $1$: a distinction invisible to $\phi^{v_1}$ yet essential for
% ordering.
% \end{example}}

Similar to the SV, PipeShapley also satisfies several promising properties for contribution attribution, including:
\begin{itemize}[leftmargin=*]
    \item \textit{Position-Aware Efficiency (PAE):}  The expected pipeline utility is fully distributed among all operators across all positions, i.e., $\sum_{i, j} \psv_{i,j}(\libraryplus, \mathcal{M}, u) = \E_{P\in \libraryplus^{\maxlen}
    | \ordersplus} \big [u(P) - u(\emptyset) \big]$. 
    \fixed{It ensures that the PipeShapley values of all operators across all positions sum exactly to the expected pipeline utility, thereby directly informing how much each operator contributes to the optimization objective, i.e., maximizing the pipeline utility.}
    \item \textit{Position-Aware Symmetry (PAS):} If two operators $o_{i}, o_{i'} \in \library$ are equivalent for the ML model at position $j$, then their contributions at position $j$ are identical.
    Formally, if $u(P_{p_j = o_i}) = u(P_{p_j = o_{i'}})$ for any pipeline $P\in \libraryplus^{\maxlen}$, then $\psv_{i, j} = \psv_{i', j}$, where $P_{p_j = o_i}$ is derived by replacing the $j$-th position of $P$ with $o_i$.
    \fixed{It ensures that operators that have identical optimization effects on pipeline utility at the same position receive equal PipeShapley values, thereby guaranteeing that contribution attribution is consistent.}
    \item \textit{Position-Aware Null Player (PANP):} If an operator $o_i$ does not change the utility of any pipeline when it is added to position $j$, its contribution at position $j$ is given zero.
    Formally, if $u(P) = u(P_{-j})$ for any pipeline $P\in \libraryplus^{\maxlen}$ with $p_j = o_i$, then $\psv_{i,j} = 0$.
    \item \textit{Linearity (LIN):} In practice, ML engineers may evaluate the performance of an ML model using multiple utility metrics. 
    PipeShapley guarantees that operators' contributions w.r.t different utility metrics are linearly additive.
    Formally, for any utility metrics $u_1, u_2$ and any weights $w_1, w_2 \in R$, we have $\psv_{i,j}(\libraryplus, \mathcal{M}, w_1\cdot u_1 + w_2 \cdot u_2) = w_1 \cdot \psv_{i,j}(\libraryplus, \mathcal{M}, u_1)+ w_2 \cdot \psv_{i,j}(\libraryplus, \mathcal{M}, u_2)$.
\end{itemize}
\begin{theorem}
    PipeShapley satisfies PAE, PAS, PANP, and LIN.
    \label{the:properties}
\end{theorem}

\begin{proof}
    See all missing proofs in our \supp{A}.
\end{proof}

\fixed{
\begin{example}[Properties of PipeShapley]
Consider a library $\library=\{o_1,o_2,o_3\}$ with placeholder $o_0$, $\maxlen=2$, and $13$ valid pipelines in $\libraryplus^{\maxlen}|\ordersplus$ (where $o_0$ cannot precede any $o_i$). The utility $u(P)$ is given by:
 
\smallskip
\centerline{
\footnotesize
\begin{tabular}{c|cccc}
$u(p_1,p_2)$ & $p_2{=}o_0$ & $p_2{=}o_1$ & $p_2{=}o_2$ & $p_2{=}o_3$ \\
\hline
$p_1{=}o_0$ & $0$    & ---    & ---    & ---    \\
$p_1{=}o_1$ & $0$    & $0$    & $0.39$ & $0.39$ \\
$p_1{=}o_2$ & $0.13$ & $0.13$ & $0.26$ & $0.39$ \\
$p_1{=}o_3$ & $0.13$ & $0.13$ & $0.26$ & $0.39$ \\
\end{tabular}
}
\smallskip
 
\noindent A direct computation yields $\psv_{1,1}{=}0.04$, $\psv_{2,1}{=}\psv_{3,1}{=}0.05$, $\psv_{1,2}{=}0$, $\psv_{2,2}{=}0.03$, $\psv_{3,2}{=}0.03$, illustrating:

\noindent\textit{(1) PAE}: the six values sum to $0.2 = \E[u(P) - u(\emptyset)]$, so each $\psv_{i,j}$ is a signed share of the pipeline's expected utility gain. Its sign indicates whether placing $o_i$ at position $j$ contributes positively or negatively, and its magnitude quantifies how much. 
Without PAE, the values would no longer aggregate to the true utility, breaking the link between $\psv_{i,j}$s and the target objective in Equation~\eqref{obj:bounded_dppc}.

\noindent\textit{(2) PAS}: at position $1$, operators $o_2$ and $o_3$ have identical effects, i.e., $u(P_{p_1=o_2})=u(P_{p_1=o_3})$ for every valid $P$ (their rows in the table coincide), so PAS forces $\psv_{2,1}=\psv_{3,1}=0.05$. A position-unaware symmetry axiom could only assign the two operators a single shared scalar, losing the resolution.

\noindent\textit{(3) PANP}: $\psv_{1,2}=0$ because $u(P_{p_2=o_1})=u(P_{p_2=o_0})$ for every valid $P$, so $o_1$ at position 2 is genuinely inert and correctly receives zero credit.
\end{example}
}

\begin{algorithm}[t]
\caption{ShapleyPipe}
\label{alg:shapleypipe}
\small
\KwIn{Raw dataset $\rawset$, learning algorithm $\learner$}
\KwOut{Pipeline $\pipeline^{\star}$}
\KwParam{Operator library $\libraryplus$, partial orders $\ordersplus$, maximum length $\maxlen$, PipeShapley approximation algorithm $\svcalc$}
Initialize an empty pipeline $P^{\star} = []$\;
\For{step $j=1$ to $\maxlen$}{
    $\psv_{1,j}(P_{<j}^{\star}), \dots, \psv_{N,j}(P_{<j}^{\star}) \gets \svcalc(\rawset, \learner, P^{\star}, j)$\;\label{alg:shapleypipe:calc_sv}
    % \IF{$\psv_{i,j}(P_{<j}^{\star}) \geq 0$ for all $i \in [N]$}
    \If{\fixed{$\phi^{PSV}_{i,j}(P^\star_{<j}) < 0$ for all $i \in [N]$}\label{alg:shapleypipe:early_term_if}}{
        Terminate the for-loop early\;\label{alg:shapleypipe:early_term_then}
    }
    Add the operator $P^{\star}_{j} = \argmax_{o_i} \psv_{i,j}(P_{<j}^{\star})$ into $P^{\star}$\;\label{alg:shapleypipe:add_to_pos_j}
}
\Return{$\pipeline^{\star}$}
\end{algorithm}

Using PipeShapley, we can fairly quantify the contributions of all operators $\library$ at each position in the pipeline, enabling a principled approach to DPPC.
Specifically, we propose a PipeShapley-based DPPC method named ShapleyPipe in Algorithm~\ref{alg:shapleypipe}.
Starting from position $1$, ShapleyPipe sequentially computes the PipeShapley values of all operators (Line~\ref{alg:shapleypipe:calc_sv}) and selects the operator with the highest value $\psv_{i,j}(P_{<j})$ to be placed at position $j$ (Line~\ref{alg:shapleypipe:add_to_pos_j}).
Note that for position $j$, since the operators at positions $1$ to $j-1$ have already been determined, we compute the conditional PipeShapley value $\psv_{i,j}(P_{<j})$ accordingly, which is defined as:
\begin{align*}
     & \psv_{i,j}(P_{<j}^{\star}) = \psv_{i,j}(\libraryplus, \mathcal{M}, u | P_{<j} = P_{<j}^{\star}) \\
    = & \mathbb{E}_{P\in \libraryplus^{\maxlen} | \ordersplus, P_{<j} = P_{<j}^{\star}} \big [\boldsymbol{1}[p_j \!=\! o_i] \!\cdot\! (u(P_{\geq j} | P_{<j}^{\star}) \!-\! u(P_{> j} | P_{<j}^{\star})) \big] \\
    = & \mathbb{E}_{P\in \libraryplus^{\maxlen} | \ordersplus, P_{<j} = P_{<j}^{\star}} \big [\boldsymbol{1}[p_j = o_i] \cdot (u(P) - u(P_{-j})) \big].
\end{align*}
If all operators have negative PipeShapley values, this indicates that their expected contributions are negative; therefore, we terminate the for loop early (Lines \ref{alg:shapleypipe:early_term_if}-\ref{alg:shapleypipe:early_term_then}).

\section{Approximating PipeShapley}
Computing PipeShapley values is intractable; therefore, in this section, we introduce how to approximate PipeShapley, thereby accelerating the ShapleyPipe algorithm for DPPC.

\begin{algorithm}[t]
\small
\caption{Monte Carlo ShapleyPipe (MC-SP)}
\label{alg:mcsp}
\KwIn{Raw dataset $\rawset$, learning algorithm $\learner$}
\KwOut{Pipeline $\pipeline^{\star}$}
\KwParam{Operator library $\libraryplus$, partial orders $\ordersplus$, maximum length $\maxlen$, number of samples $\noofsamples$}
\Return{$\kw{ShapleyPipe}(\rawset, \learner, \kw{MCApprox})$}\;
\Procedure{$\kw{MCApprox}(\rawset, \learner, P^{\star}, j)$\label{alg:mcsp:mcapprox-def}}{
    Initialize $\phi_{i,j}(P^{\star}_{<j}) \leftarrow 0$ for all $o_i \in \library$\;
    
    \For{$o_i \in \library$}{
        \For{$t \leftarrow 1$ \KwTo $\noofsamples$}{
            Sample $P \in \libraryplus^{\maxlen - j} \mid \ordersplus$ s.t. $P_{< j} = P^{\star}_{< j}$ and $P_j = o_i$\;\label{alg:mcsp:sample-p-satisfy-constraint}
            
            $\phi_{i,j}(P^{\star}_{<j}) \leftarrow \phi_{i,j}(P^{\star}_{<j}) 
            + \big(u(P) - u(P_{-j})\big) / \noofsamples$\;
        }
    }
\Return{$\phi_{1,j}(P^{\star}_{<j}),\dots,\phi_{N,j}(P^{\star}_{<j})$}\;\label{alg:mcsp:mcapprox-return}
}
\end{algorithm}

\subsection{Monte Carlo Approximation}
The Monte Carlo method approximates the expectation of a random variable using random samples. 
Since PipeShapley $\psv_{i,j}$ is formulated as the expectation of the random variable $[\boldsymbol{1}[p_j = o_i] \cdot (u(P_{\geq j}) - u(P_{> j}))$ over the space $\{P \mid P\in \libraryplus^{\maxlen}| \ordersplus\}$, the Monte Carlo method can be used to approximately compute PipeShapley.

Therefore, we propose Monte Carlo ShapleyPipe (MC-SP, see Algorithm~\ref{alg:mcsp}). 
In MC-SP, we approximate PipeShapley values using Monte Carlo approximation $\kw{MCApprox}$ (Lines \ref{alg:mcsp:mcapprox-def}--\ref{alg:mcsp:mcapprox-return}). 
For each operator $o_i$, we perform $\noofsamples$ random samplings to generate operator pipelines $P$ satisfying the conditions $P_{< j} = P^{\star}_{< j}$ and $p_j = o_i$ (Line \ref{alg:mcsp:sample-p-satisfy-constraint}), and compute the marginal contribution $(u(P) - u(P_{-j}))$ of operator $o_i$ at position $j$. 
Consequently, the PipeShapley values for all operators at all positions can be computed with $O(N \cdot \noofsamples \cdot \maxlen)$ pipeline utility evaluations, which is a polynomial complexity.

\begin{figure}[t]
\centering
    \makeatletter
    \if@twocolumn
      \includegraphics[width=\columnwidth]{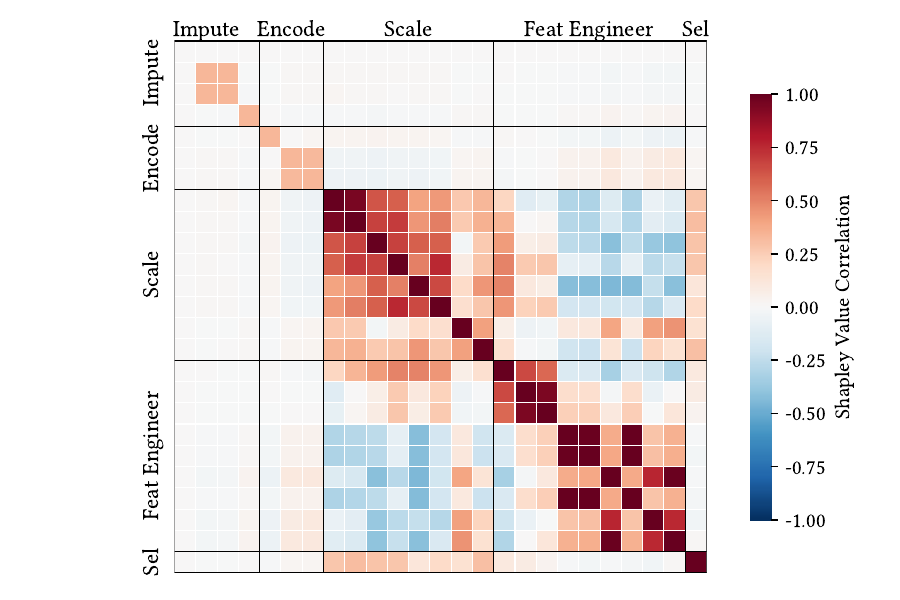}
    \else
      \includegraphics[width=0.7\textwidth]{Figures/correlation/similarity_matrix_split_j_mean.pdf}
    \fi
    \makeatother
% \includegraphics[width=\linewidth]{Figures/correlation/similarity_matrix_split_j_mean.pdf}
% \caption{The matrix shows the Pearson correlation coefficient of PipeShapley values across all positions for each pair of operators. Operators within the same category tend to exhibit stronger Pearson correlation.}
\caption{Pearson correlation coefficient of PipeShapley values of each pair of operators across positions averaged over the 27 datasets. Operators within the same category tend to exhibit stronger Pearson correlation.}
\label{fig:coherence_mean}
\end{figure}

\fixed{
\subsection{Category Coherence Assumption}
However, we consider the complexity of MC-SP to remain relatively high. 
To further reduce its computational complexity, we aim to investigate the patterns in the contribution distributions of different operators across different positions, and leverage these patterns to design a more efficient PipeShapley approximation algorithm.

Inspired by CtxPipe~\cite{gao2024ctxpipe}, we categorize all operators into five categories according to their functionality: Imputer, Encoder, Scaler, Feature Engineering, and Feature Selection. 
Based on this categorization, we propose the Category Coherence Assumption: \textbf{operators within the same category exhibit coherent contribution patterns.} 
To validate this assumption, for each pair of operators $o_i, o_{i'}$, we calculate their PipeShapley values across all positions, i.e., $[\psv_{i,1}, \dots, \psv_{i, \maxlen}]$ for $o_i$ and $[\psv_{i',1}, \dots, \psv_{i', \maxlen}]$ for $o_{i'}$.
Then, in Figure~\ref{fig:coherence_mean}, we report the Pearson correlation coefficient (PCC) of the two operators' PipeShapley values across positions, i.e., 
$$\text{PCC}([\psv_{i,1}, \dots, \psv_{i, \maxlen}], [\psv_{i',1}, \dots, \psv_{i', \maxlen}]).$$ 
We can observe that two operators from the same category consistently exhibit a much higher PCC than those from different categories, validating our assumption.
% To validate this assumption, we then compute the Pearson correlations between the PipeShapley values of all operators and visualize the results as an $N \times N$ matrix in Figure~\ref{fig:coherence_mean}. 
% The results show that operators within the same category indeed exhibit stronger positive correlations in their PipeShapley values, which implies coherent contribution patterns.

The validity of the Category Coherence Assumption implies that when an operator $o_i$ exhibits a high contribution at position $j$, this provides an informative signal that other operators within the same category are also likely to have high contributions at position $j$. 
Therefore, instead of uniformly estimating the PipeShapley values of all operators in MC-SP, we can perform hierarchical exploration over operator categories and share contribution signals within each category, thereby reducing the number of pipeline utility evaluations.
}

\begin{algorithm}[t]
\small
\caption{Hierarchical ShapleyPipe (H-SP)}
\label{alg:hsp}
\KwIn{Raw dataset $\rawset$, learning algorithm $\learner$}
\KwOut{Pipeline $\pipeline^{\star}$}
\KwParam{Operator library $\libraryplus$, partial orders $\ordersplus$, maximum length $\maxlen$, category library $\categories$, number of samples $\noofsamples$, initial MAB parameters $(\mabscore_{i}, t_{i}, t_{\categories_k})$ for all $o_i \in \library$ and for all $\categories_k \in \categories$}
Get category pipeline $\catepipe = \kw{GetCatePipe}(\rawset, \learner)$\;\label{alg:hsp:stage-a-expression}
\Return{$\kw{ShapleyPipe}(\rawset, \learner, \kw{HierApprox}(\cdot | \catepipe))$}\;\label{alg:hsp:stage-b-expression}
\Procedure{$\kw{GetCatePipe}(\rawset, \learner)$\label{alg:hsp:get-cate-def}}{
    Initialize $\catepipe = []$\;
    \For{$j = 1$ \KwTo $\maxlen$}{
        Initialize $\phi_{\categories_k,j} = 0$ for all $\categories_k \in \categories$\;
        \For{$\categories_k \in \categories$}{
            \For{$t = 1$ \KwTo $\noofsamples$\label{alg:hsp:cate-n_sample-start}}{
                Sample category pipeline $\catepipe' \in \categories_{+}^{\maxlen}$ s.t. $\catepipe'_{<j} = \catepipe_{<j}, \catepipe'_{j} = \categories_{k}$\;
                $P \gets []$\;
                \For{$j'=1$ \KwTo $|\catepipe'|$}{
                Add $\argmax_{o_i\in \catepipe_{j'}} (\frac{\mabscore_{i}}{t_{i}} + \sqrt{\frac{2 \ln t_{\catepipe_{j'}}}{t_{i}}})$ to $P_{j'}$ s.t. $\ordersplus$\;\label{alg:hsp:stage-a-add-mab-oi-to-P}
                }
                $o_{i^{\star}} \gets P_{j}$\;\label{alg:hsp:oi_star-leftarr-P_j}
                $x_{i^{\star}} \gets x_{i^{\star}} + (u(P) - u(P_{-j})), t_{i^{\star}} \gets t_{i^{\star}} + 1, t_{\categories_k} \gets t_{\categories_k} + 1$\;\label{alg:hsp:update-mab}
                $\phi_{\categories_k,j} \gets \phi_{\categories_k,j} + (u(P) - u(P_{-j})) / \noofsamples$\label{alg:hsp:accum-phi_c_k_j}
                }
            }
        Get optimal category $\catepipe_j = \argmax_{\categories_k} \phi_{\categories_k,j}$\;\label{alg:hsp:get-CP_j}
        }
    \Return{$\catepipe$}\;\label{alg:hsp:get-cate-return}
}
\Procedure{$\kw{HierApprox}(\rawset, \learner, P^{\star}, j | \catepipe)$\label{alg:hsp:hierapprox-def}}{
    Initialize $\phi_{i,j}(P^{\star}_{<j}) = 0$ for all $o_i \in \library$\;
    \For{$o_i \in \catepipe_j$}{
        \For{$t = 1$ \KwTo $\noofsamples$\label{alg:hsp:n_sample-iter-start}}{
            Sample $P \in \libraryplus^{\maxlen} \mid \ordersplus$ s.t. $P_{\leq j} = P^{\star}_{< j} \oplus [o_i]$\;            
            $\phi_{i,j}(P^{\star}_{<j}) \leftarrow \phi_{i,j}(P^{\star}_{<j}) 
            + \big(u(P) - u(P_{-j})\big) / \noofsamples$\;\label{alg:hsp:stage-b-accum-phi_ij}
        }
    }
    \Return{$\phi_{1,j}(P^{\star}_{<j}),\dots,\phi_{N,j}(P^{\star}_{<j})$}\;\label{alg:hsp:hierapprox-return}
}
\end{algorithm}

\subsection{Hierarchical Approximation}

Motivated by the Category Coherence Assumption, we propose \textbf{Hierarchical ShapleyPipe (H-SP)} in Algorithm~\ref{alg:hsp}, a two-stage approximation algorithm that exploits category-level structure to reduce the number of pipeline utility evaluations.
H-SP proceeds as follows: in \textbf{Stage A} (\emph{Category Pipeline Construction}, Line~\ref{alg:hsp:stage-a-expression}), we first determine, for each position $j$, which operator category is most likely to yield the highest PipeShapley value, producing a \emph{category pipeline} $\catepipe$; in \textbf{Stage B} (\emph{Category-Constrained Operator Pipeline Construction}, Line~\ref{alg:hsp:stage-b-expression}), we restrict the operator-level PipeShapley computation at each position $j$ to only the operators belonging to $\catepipe_j$, the category selected in Stage A.

\myparagraph{Stage A: Category Pipeline Construction.}
The goal of Stage A is to identify, for each pipeline position $j$, the operator category most likely to contain the best operator.
To this end, we define the \emph{category-level PipeShapley value} of category $\categories_k$ at position $j$ as the maximum operator-level PipeShapley value within that category:
$$\phi_{\categories_k, j} \;=\; \max_{o_i \in \categories_k}\; \psv_{i,j}(P^{\star}_{<j}).$$
% \begin{equation*}
%     \phi_{\categories_k, j} \;=\; \max_{o_i \in \categories_k}\; \psv_{i,j}(P^{\star}_{<j}).
%     \label{eq:cat-psv}
% \end{equation*}
Computing $\phi_{\categories_k,j}$ exactly would still require evaluating every operator in $\categories_k$, which defeats the purpose of approximation.
Instead, we leverage a Multi-Armed Bandit (MAB) strategy to \emph{identify the operator with the highest marginal contribution within each category}, using it as a proxy for the category-level value.
Concretely, for each position $j$ and each category $\categories_k$, we draw $\noofsamples$ category-level pipeline samples (Lines~\ref{alg:hsp:cate-n_sample-start}--\ref{alg:hsp:stage-a-add-mab-oi-to-P}).
In each sample, a category pipeline $\catepipe'$ is drawn such that $\catepipe'_j = \categories_k$ and $\catepipe'_{<j} = \catepipe_{<j}$; the concrete operator at each position $j'$ in the sampled category pipeline is then resolved by the UCB selection rule~\cite{auer2002finite}:
\[
    o_{j'} = \argmax_{o_i \in \catepipe'_{j'}} \left(\frac{\mabscore_i}{t_i} + \sqrt{\frac{2\ln t_{\catepipe_{j'}}}{t_i}}\right),
\]
where $\mabscore_i$ is the cumulative reward of operator $o_i$, $t_i$ is the number of times $o_i$ has been selected, and $t_{\catepipe'_{j'}}$ is the total number of pulls for category $\catepipe_{j'}$ (the $j'$-th category in category pipeline $\catepipe$).
The selected operator $o_{i^\star}$ at position $j$ (Line~\ref{alg:hsp:oi_star-leftarr-P_j}) receives a marginal contribution reward $u(P) - u(P_{-j})$, and its MAB statistics are updated accordingly (Line~\ref{alg:hsp:update-mab}).
The average marginal contribution across all samples is used to estimate $\phi_{\categories_k,j}$ (Line~\ref{alg:hsp:accum-phi_c_k_j}).
After iterating over all categories, the category with the highest estimated value is selected as $\catepipe_j$ (Line~\ref{alg:hsp:get-CP_j}).
Crucially, the initial MAB parameters $(\mabscore_i, t_i, t_{\categories_k})$ are not set to zero from scratch; they can be pre-trained from historical data of previously constructed data preparation pipelines, providing a warm start that substantially reduces the number of samples required for accurate category selection.
 
\myparagraph{Stage B: Category-Constrained Construction of Operator \\Pipeline.}
Given the category pipeline $\catepipe$ produced by Stage A, Stage B computes operator-level PipeShapley values restricted to each selected category.
The procedure \kw{HierApprox} (Lines~\ref{alg:hsp:hierapprox-def}--\ref{alg:hsp:hierapprox-return}) implements Stage B: at position $j$, we compute $\psv_{i,j}(P^{\star}_{<j})$ only for operators $o_i \in \catepipe_j$, the category identified in Stage A.
For each such operator, we draw $\noofsamples$ pipeline samples conditioned on $P_{\leq j} = P^{\star}_{<j} \oplus [o_i]$, and estimate its marginal contribution by averaging $u(P) - u(P_{-j})$ across samples (Lines~\ref{alg:hsp:n_sample-iter-start}--\ref{alg:hsp:stage-b-accum-phi_ij}).
By restricting operator-level evaluation to a single category per position, Stage B avoids the cost of evaluating operators from irrelevant categories, significantly reducing the total number of pipeline utility evaluations without sacrificing the quality of the final pipeline.

\fixed{
\myparagraph{Complexity Comparison.}
Table~\ref{tab:complexity} summarizes the pipeline utility evaluation complexity of MC-SP and H-SP.
% Here $N$ is the total number of operators, $K$ is the number of categories, and $\max_{k}|\categories_k|$ is the size of the largest category.
H-SP offers a clear efficiency advantage when the operator library is large but operators are distributed across many categories of moderate size, i.e., when $K + \max_{k}|\categories_k| \ll N$.
In particular, if operators are distributed roughly uniformly across $K$ categories, then $\max_k |\categories_k| \approx N/K$, and H-SP reduces the evaluation cost by a factor of approximately $N / (K + N/K)$ relative to MC-SP, a saving that becomes substantial as $K$ grows.
In contrast, when $K$ is small and one dominant category contains most operators ($\max_k|\categories_k| \approx N$), the savings are limited, and MC-SP may be preferred for its simplicity.
}

\begin{table}[t]
\centering
\caption{Comparison of pipeline utility evaluation complexity. $N$ is the total number of operators, $K$ is the number of categories, and $\max_{k}|\categories_k|$ is the size of the largest category.}
\label{tab:complexity}
\small
\begin{tabular}{lc}
\toprule
\textbf{Algorithm} & \textbf{\# Pipeline Utility Evaluations} \\
\midrule
MC-SP & $O(N \cdot \noofsamples \cdot \maxlen)$ \\
H-SP  & $O\!\left((K + \max_{k}|\categories_k|) \cdot \noofsamples \cdot \maxlen\right)$ \\
\bottomrule
\end{tabular}
\end{table}
 
\fixed{
\begin{example}[Hierarchical ShapleyPipe]
We illustrate how H-SP exploits operator categorization to compute PipeShapley values in two stages, using our actual operator library of $N=25$ operators grouped into $K=5$ categories (Imputer, Encoding, Scaling, Engineering, Selection) with $4, 3, 8, 9, 1$ operators respectively, pipeline length $\maxlen = 5$, and $\noofsamples = 10$ samples per estimation.
In Stage A, H-SP estimates the category-level PipeShapley value $\phi_{\categories_k, j}$ for each of the $K=5$ categories at each of the $5$ positions and picks the winning one, producing a category pipeline $\catepipe$ of length $5$.
In Stage B, H-SP computes operator-level PipeShapley values only within $\catepipe_j$ at each position, evaluating at most $\max_k|\categories_k| = 9$ operators per position instead of all $25$, and selects the operator with the highest value as $\pipeline^{\star}_j$.
In total, H-SP performs at most $(K + \max_k|\categories_k|) \cdot \maxlen \cdot \noofsamples = (5+9) \cdot 5 \cdot 10 = 700$ pipeline utility evaluations, compared to $N \cdot \maxlen \cdot \noofsamples = 25 \cdot 5 \cdot 10 = 1250$ for MC-SP, nearly halving the cost.
\end{example}
}

% \input{Sections/methods}

% \section{Adaptive Advisor Triggering }\label{sec:optimization}
% \input{optimization}

\section{Experiments}\label{sec:exp}

\label{sec:experiments}

In this section, we conduct comprehensive experiments to evaluate ShapleyPipe by addressing the following research questions:
\begin{itemize}[leftmargin=*]
    \item \textbf{RQ1 (Effectiveness \& Efficiency):} How effective and efficient is ShapleyPipe compared against SOTA methods for DPPC?
    % \item \textbf{RQ2 (Efficiency):} How efficient is ShapleyPipe compared against SOTA methods for DPPC?
    \item \textbf{RQ2 (Ablation):} What is the individual contribution of each core component in H-SP?
    \item \textbf{RQ3 (Sensitivity):} How robust is ShapleyPipe to variations in hyperparameter settings?
    % \item \textbf{RQ4 (Extensibility)}: Does the Category Coherence Assumption survive a larger and more diverse operator library?
\end{itemize}

% \paragraph{Baselines.}
% We compare ShapleyPipe against 11 representative baselines spanning different categories:

% \begin{itemize}[leftmargin=*]
%     \item \textbf{Search-based methods:} \textit{ES*} (exhaustive search with pruning), \textit{Random} (random search with fixed budget)
%     \item \textbf{AutoML systems:} \textit{TPOT}~\cite{tpot}, \textit{Auto-sklearn (AUSK)}~\cite{autosklearn}
%     \item \textbf{Data-centric methods:} \textit{CtxPipe}~\cite{ctxpipe}, \textit{HAI-AI}~\cite{hai-ai}, \textit{CatDB}~\cite{catdb}
%     \item \textbf{LLM-based methods:} \textit{ML-Agent}~\cite{mlagent}, \textit{GPT4o-ZS} (zero-shot), \textit{GPT4o-FS} (few-shot)
%     \item \textbf{Ablation variants:} \textit{LOO} (Leave-One-Out without position awareness), \textit{ShapleyPipe} (PipeShapley selection without hierarchical decomposition)
% \end{itemize}

% =============================================================================
% RQ1: EFFECTIVENESS
% =============================================================================

\begin{table}[t]
\centering
\caption{Operator categorization based on functionality.}
\label{tab:operators}
\setlength{\tabcolsep}{3pt}
\small
\begin{tabular}{p{3cm}rp{11.5cm}}
\toprule
\textbf{Category} & \textbf{Count} & \textbf{Operator (ID)} \\
\midrule
Imputer & 4 & Categorical Most Frequent (0), NumMean (1), NumMedian (2), Num Most Frequent (3) \\
Encoding & 3 &  Numeric Data (4), Label (5), One-Hot (6) \\
Scaling & 8 & MinMax (7), MaxAbs (8), Robust (9), Standard (10), Quantile Transform (11), Power Transform (12), Normalizer (13), K-Bins Discretizer (14) \\
Feature Engineering & 9 & Polynomial (15), Interaction (16), PCA Auto (17), PCA Lapack (18), PCA Arpack (19), Incremental PCA (20), Kernel PCA (21), Truncated SVD (22), Tree Embedding (23) \\
Feature Selection & 1 & Variance Threshold (24) \\
% \midrule
% \textbf{Total} & \textbf{25} & -- \\
\bottomrule
\end{tabular}
\end{table}

\begin{table*}[t]
\centering
% \caption{\textbf{Main Performance Comparison ($\protect\maxlen=6, N=25$).} Excluding the ES* method, H-SP achieves the highest mean accuracy across 27 datasets. It outperforms baselines and approaches the empirical upper bound (ES*), demonstrating the effectiveness of game-theoretical guidance. Best results are bolded. ML-Agent results marked N/A due to context length limitations.}
\caption{\textbf{Main Performance Comparison ($\protect\maxlen=6, N=25$).} Bold and underline mark the best and runner-up per row, excluding ES*, a 10{,}000-evaluation reference rather than a practical competitor. Among practical methods, H-SP achieves the highest mean accuracy across the 27 datasets. ML-Agent fails on three datasets due to context-length limits (N/A).}
\label{tab:main_results}
\setlength{\tabcolsep}{1.8pt}
\resizebox{\textwidth}{!}{%
\begin{tabular}{l|cc|cccc|ccc|cccc|cc|c}
\toprule
\textbf{Dataset} & \textbf{ES*} & \textbf{Rand} & \textbf{TPOT} & \textbf{AUSK} & \fixed{\textbf{SAGA}} & \fixed{\textbf{DP}} & \fixed{\textbf{L2C}} & \textbf{CtxPipe} & \textbf{HAIPipe} & \textbf{CatDB} & \textbf{ML-Agent} & \textbf{GPT4o-ZS} & \textbf{GPT4o-FS} & \textbf{LOO} & \textbf{MC-SP} & \textbf{H-SP} \\
\midrule             % ES     Rand     TPOT              AUSK    SAGA              DP                L2C               CtxPipe          HAI-AI   CatDB            ML-Agent       GPTzs  GPTfs    LOO                MC-SP              H-SP
abalone              & 0.287 & 0.243 & \textbf{0.289}  & 0.243 & 0.255           & 0.271           & 0.273           &\underline{0.287}& 0.260 & 0.258           & 0.273        & 0.267 & 0.250 & 0.245           & 0.277            & 0.278 \\
ada\_prior           & 0.857 & 0.844 & 0.823           & 0.832 & 0.833           & \textbf{0.846}  & 0.791           & 0.818           & 0.801 & 0.835           & 0.838        & 0.839 & 0.831 & 0.821           & 0.844            & \underline{0.845} \\
avila                & 0.916 & 0.598 & 0.593           & 0.539 & 0.636           & 0.633           & 0.678           & 0.759           & 0.630 & 0.739           & 0.570        & 0.609 & 0.679 & 0.754           &\underline{0.840} & \textbf{0.941} \\
connect-4            & 0.792 & 0.671 & 0.658           & 0.660 & 0.758           & 0.702           & 0.747           & 0.763           & 0.775 & 0.759           & N/A          & 0.713 & 0.657 &\underline{0.793}&\textbf{0.794}    & \textbf{0.794} \\
eeg                  & 0.861 & 0.658 & 0.594           & 0.647 & 0.658           & 0.683           & 0.848           & 0.740           & 0.556 & 0.529           & 0.556        & 0.623 & 0.834 & 0.740           &\underline{0.856} & \textbf{0.864} \\
google               & 0.672 & 0.627 & 0.602           & 0.558 & 0.596           & 0.661           & 0.641           & 0.590           & 0.550 & 0.602           & N/A          & 0.548 & 0.648 &\underline{0.675}&\textbf{0.677}    & 0.665 \\
house\_prices        & 0.955 & 0.912 &\underline{0.938}& 0.912 & 0.913           & \textbf{0.952}  & 0.935           & 0.818           & 0.928 & 0.906           & 0.926        & 0.904 & 0.928 & 0.935           &0.928             & \underline{0.938} \\
jungle\_chess        & 0.861 & 0.669 & 0.682           & 0.683 & 0.745           & 0.687           & 0.814           & \textbf{0.861}           & 0.760 & 0.747           & N/A          & 0.760 & 0.807 &\textbf{0.861}   &\underline{0.856}             & \textbf{0.861} \\
micro                & 0.634 & 0.579 & 0.558           & 0.565 & 0.556           & 0.593           & 0.613           & 0.605           & 0.633 & 0.590           & 0.351        & 0.567 & \underline{0.639} &0.635&0.633             & \textbf{0.643} \\
mozilla4             & 0.940 & 0.922 & 0.890           & 0.849 & 0.932           & 0.927           & \underline{0.939}           & \textbf{0.940}  & 0.870 & 0.820           & 0.856        & 0.867 & 0.931 & \textbf{0.940}  &0.937 & 0.934 \\
obesity              & 0.927 & 0.841 & \textbf{0.901}           & 0.695 & 0.751           & 0.874           & 0.790           & 0.868           & 0.768 & 0.724           & 0.676        & 0.851 & 0.816 &\underline{0.891}&0.882             & \textbf{0.901} \\
page-blocks          & 0.973 & 0.959 & 0.967           & 0.954 & 0.849           &\underline{0.973}&\underline{0.973}& 0.965           & 0.935 & 0.953           &\textbf{0.979}& 0.955 & 0.908 & 0.968           &0.968             & 0.970 \\
pbcseq               & 0.866 & 0.730 & 0.715           & 0.681 & \textbf{0.866}  & 0.725           & 0.712           &\underline{0.805}& 0.733 & 0.644           & 0.756        & 0.704 & 0.694 & 0.758           &0.756             & 0.784 \\
pol                  & 0.977 & 0.879 & 0.894           & 0.885 & 0.888           & 0.916           & 0.926           & 0.949           & 0.916 & 0.947           & 0.843        & 0.921 & 0.888 &\underline{0.981}&0.949             & \textbf{0.982} \\
run\_or\_walk        & 0.990 & 0.829 & 0.826           & 0.710 & 0.832           & 0.912           & 0.920           & 0.956           & 0.915 & 0.795           & 0.662        & 0.718 & 0.859 & 0.972           &\underline{0.973} & \textbf{0.984} \\
shuttle              & 1.000 & 0.996 & 0.933           & 0.926 & 0.405           &\underline{0.999}& 0.986           & \textbf{1.000}  & 0.951 &\underline{0.999}& 0.957        & 0.994 & 0.984 & \textbf{1.000}  &0.998             & \textbf{1.000} \\
uscensus             & 0.854 & 0.840 & 0.828           & 0.825 & 0.835           & \textbf{0.852}  & 0.826           & 0.845           & 0.807 & \underline{0.851}           & 0.749        & 0.841 & \underline{0.851} & \textbf{0.852}  &0.846 & 0.846 \\
wall-robot           & 0.962 & 0.872 & 0.754           & 0.692 & 0.841           & 0.913           & 0.904           & 0.946           & 0.896 & 0.898           & 0.816        & 0.812 & 0.941 &\underline{0.947}&\textbf{0.952}    & \textbf{0.952} \\
cmc                  & 0.635 & 0.554 & 0.534           & 0.486 & 0.550           & 0.510           & 0.581           & 0.534           & 0.549 & 0.557           & 0.463        & 0.568 & 0.568 & 0.574           &\underline{0.595} & \textbf{0.615} \\
diabetes             & 0.909 & 0.831 & 0.831           & 0.769 & 0.772           & 0.765           & 0.844           & 0.766           & 0.773 & 0.750           & 0.677        & 0.766 & 0.831 & 0.818           &\underline{0.857} & \textbf{0.870} \\
jm1                  & 0.827 & 0.813 & 0.812           & 0.809 & 0.817           & 0.807           &\underline{0.818}& 0.809           & 0.809 & 0.797           & 0.786        & 0.812 & 0.812 & 0.811           &\underline{0.818} & \textbf{0.824} \\
kc1                  & 0.896 & 0.839 & 0.867           & 0.858 & 0.861           & 0.865           &\underline{0.872}&\underline{0.872}& 0.865 & 0.844           & 0.861        & 0.834 & 0.867 &\underline{0.872}&\underline{0.872} & \textbf{0.882} \\
kc2                  & 0.925 & 0.792 & 0.811           & 0.826 &0.856& 0.846           & 0.830           & 0.792           & \underline{0.857} & 0.815           & 0.753        & 0.830 & 0.736 & 0.811           &0.830             & \textbf{0.868} \\
pc4                  & 0.959 & 0.877 &\underline{0.925}& 0.898 & 0.904           & 0.900           & 0.924           & 0.904           & 0.914 & 0.897           & 0.846        & 0.890 & 0.890 & 0.904           &\textbf{0.938}    & \textbf{0.938} \\
phoneme              & 0.913 & 0.754 & 0.832           & 0.760 & 0.786           & 0.756           & 0.867           & 0.824           & 0.772 & 0.750           & 0.795        & 0.751 & 0.823 & 0.867           &\underline{0.869} & \textbf{0.882} \\
qsar-biodeg          & 0.934 & 0.764 & 0.849           & 0.882 & 0.879           & 0.853           & 0.887 & 0.858           & 0.867 & 0.858           & 0.820        & 0.868 & 0.821 & 0.849           & \textbf{0.915}   & \underline{0.906} \\
vehicle              & 0.953 & 0.694 & 0.847           & 0.778 & 0.787           & 0.740           & 0.847           & 0.788           & 0.812 & 0.701           & 0.783        & 0.741 & 0.871 & \underline{0.894}         &\textbf{0.906} & \textbf{0.906} \\ \midrule
\textbf{Avg. Acc. ↑} & 0.862 & 0.762 & 0.769           & 0.738 & 0.754           & 0.784           & 0.807           & 0.802           & 0.774 & 0.762           & 0.733        & 0.761 & 0.791 & 0.821           & \underline{0.836}& \textbf{0.847} \\
\textbf{Avg. Rank ↓} & --    & 10.2  & 8.9             & 12.0  & 9.1             &   7.5           &   5.7           &   6.5           &   8.8 &  10.2           & 12.1         &  10.0 &   8.0 &   4.7           & \underline{3.2}  & \textbf{2.0} \\
\textbf{Avg. Prec. ↑}& --    & 0.757 & 0.795           & 0.721 & 0.670           & 0.711           & 0.804           & 0.800           & 0.762 & 0.712           & 0.719        & 0.772 & 0.767 & 0.810           & 0.832 & 0.840 \\
\textbf{Avg. Recall ↑}& --   & 0.762 & 0.803           & 0.738 & 0.663           & 0.671           & 0.809           & 0.808           & 0.773 & 0.741           & 0.720        & 0.779 & 0.782 & 0.819           & 0.839 & 0.845 \\
\textbf{Avg. F1   ↑} & --    & 0.759 & 0.792           & 0.712 & 0.663           & 0.669           & 0.798           & 0.800           & 0.755 & 0.713           & 0.709        & 0.768 & 0.766 & 0.809           & 0.828 & 0.837 \\
\textbf{Avg. Time (s) ↓}&$\approx$22k& 48.2  & 342.8           & 482.9 & 1384.5          & 1029.8          & 97.6            & 102.4           & 89.2  & 308.6           & 539.6        & 45.3  & 68.4  & 261.8           & 2035.8& 1294.3 \\ 
\textbf{Avg. \#Evals ↓}&10k  & 20    & 316             & 537   & 939             & 2000            & 1024            & 6               & 6     & 16              & 34           & 5     & 5     & 150             & 2845  & 2136 \\ 
\bottomrule
\end{tabular}
} % resizebox
\end{table*}

\subsection{Experimental Setup}
\label{subsec:setup}

\myparagraph{Datasets.}
Our main experiments use 27 real-world classification datasets from OpenML~\cite{vanschoren2014openml}, the 18 datasets used by DiffPrep~\cite{li2023diffprep} and 9 datasets from the CC18 suite~\cite{bischl2021openml}, which are widely used for the DPPC task and span diverse domains, scales, and data-quality issues (e.g., missing values and class imbalance). To assess scalability, we additionally evaluate on three large-scale OpenML datasets: HIGGS, Covertype, and Airlines. Each dataset is split into 80\% for training and 20\% for testing. Dataset statistics are provided in \supp{C.1}.

\myparagraph{Learning Algorithms $\mathcal{A}$.}
Following CtxPipe~\cite{gao2024ctxpipe}, we adopt Logistic Regression (LR) as the default learning algorithm. LR is chosen for two reasons: it is widely used in real-world engineering practice, and it is sensitive to operator selection, making it effective at exposing the effectiveness differences across DPPC methods. To verify that our conclusions generalize, we further evaluate three additional algorithms: Random Forest, XGBoost, and MLP. All models are implemented using the scikit-learn and xgboost libraries with fixed random seeds.

\myparagraph{Baselines.}
We consider baselines spanning search-based, RL-based, and LLM-based methods. \textbf{Search-based methods:} Rand (random sampling); TPOT~\cite{olson2016evaluation}; \fixed{SAGA~\cite{siddiqi2023saga}}; Auto-sklearn (AUSK)~\cite{autosklearn}; \fixed{DiffPrep (DP)~\cite{li2023diffprep}}; ES* (approximated Exhaustive Search)~\cite{gao2024ctxpipe}, which approximates exhaustive search with 10{,}000 pipeline evaluations and serves as a high-budget reference rather than a practical competitor; Leave-One-Out (LOO), which constructs the pipeline position by position and adds the operator with the largest utility improvement at each position. 
\textbf{RL-based methods:} \fixed{Learn2Clean (L2C)~\cite{berti2019learn2clean}}; CtxPipe~\cite{gao2024ctxpipe}; HAIPipe~\cite{chen2023haipipe}.
\textbf{LLM-based methods:} AutoML-Agent (ML-Agent)~\cite{automlagent}; CatDB~\cite{catdb}; zero-shot GPT-4o (GPT4o-ZS); GPT-4o with few-shot learning (GPT4o-FS).

% \begin{itemize}[leftmargin=*]
%     \item \textbf{Search-based methods:} Rand (random sampling); TPOT~\cite{olson2016evaluation}, \fixed{SAGA~\cite{siddiqi2023saga}}; Auto-Sklearn (AUSK)~\cite{autosklearn}; \fixed{DiffPrep (DP)~\cite{li2023diffprep}}; ES* (approximated Exhaustive Search)~\cite{gao2024ctxpipe}, which approximates exhaustive search with 10{,}000 pipeline evaluations and serves as a strong empirical upper bound; Leave-One-Out (LOO), which constructs the pipeline position by position and adds the operator with the largest utility improvement at each position.
%     \item \textbf{RL-based methods:} \fixed{Learn2Clean (L2C)~\cite{berti2019learn2clean}}; CtxPipe~\cite{gao2024ctxpipe}; HAI-AI~\cite{chen2023haipipe};
%     \item \textbf{LLM-based methods:} ML-Agent~\cite{automlagent}, CatDB~\cite{catdb}; zero-shot GPT-4o (GPT4o-ZS); GPT-4o with few-shot learning (GPT4o-FS).
% \end{itemize}

% \myparagraph{Operator Library.}
% We use a library of 25 common operators, consistent with the implementation of CtxPipe~\cite{gao2024ctxpipe}. All operators are implemented using scikit-learn with their default parameters.

% \myparagraph{Operator Categories.}
% As shown in Table~\ref{tab:operators}, we group the 25 operators into 5 categories based on functionality, consistent with the categorization schemes used in prior studies~\cite{berti2019learn2clean, chen2023haipipe, gao2024ctxpipe, heffetz2020deepline,  siddiqi2023saga}. 

\myparagraph{Operator Library and Categories.}
The default library contains 25 operators, consistent with the implementation of CtxPipe~\cite{gao2024ctxpipe}. The implementation is based on scikit-learn. As listed in Table~\ref{tab:operators}, we group the 25 operators into 5 categories based on functionality, consistent with the categorization schemes used in prior studies~\cite{berti2019learn2clean, chen2023haipipe, gao2024ctxpipe, heffetz2020deepline,  siddiqi2023saga}. Besides, we also evaluate an extended library of 32 operators in seven categories, described in \supp{D}.

\fixed{
\myparagraph{Partial Orders.}
% When constructing pipelines, we impose basic partial-order constraints that reflect common real-world practice:
We impose constraints that reflect common practice:
\(\text{Imputation} \prec \text{Encoding} \prec
\{\text{Scaling},\;\text{Engineering},\;\text{Selection}\}.\)
Encoding must follow imputation (since most encoders cannot handle missing values), and scaling, feature engineering, and feature selection must follow encoding, while these three are mutually unordered. These constraints substantially reduce the search space: at $\maxlen = 6$, they shrink the number of candidate pipelines from $3.09\times10^{8}$ to $5.14\times10^{7}$.
}
\supp{C.2} gives the operator-level constraints.

\myparagraph{Evaluation Metrics.}
% We report classification accuracy and runtime to measure the effectiveness and efficiency of DPPC methods, respectively. As auxiliary metrics, we report macro-averaged precision, recall, and F1 for effectiveness, and the number of pipeline evaluations for efficiency. We adopt macro-averaging so that all classes are weighted equally, providing insights beyond overall accuracy.
We report classification accuracy on the test split for effectiveness and wall-clock time for efficiency. As auxiliary metrics, we report macro-averaged precision, recall and F1, which weight all classes equally, and the number of
trained pipelines. Following~\cite{demsar2006statistical,benavoli2016should}, we compare H-SP with each baseline by a two-sided Wilcoxon signed-rank test over datasets on per-dataset means over the 5 seeds, apply Holm correction across baselines, and also report the Friedman test over all methods that
complete every dataset.

\myparagraph{Large-dataset protocol.}\label{sec:large-dataset_protocol}
% For the large-dataset experiments, H-SP, MC-SP, and the search-based baselines construct pipelines using the same uniformly sampled subset of 10,000 training records. Each selected pipeline is then retrained and evaluated on the full dataset, whereas the RL- and LLM-based baselines follow their original protocols and do not use the 10K subset.
H-SP, MC-SP, and the search-based baselines construct pipelines using the same uniformly sampled subset of 10,000 training records. The selected pipelines are subsequently evaluated at full data scale. The RL- and LLM-based methods follow their respective protocols and do not use this common subset.

\myparagraph{Runtime accounting.}
% The reported CtxPipe and HAI-AI times exclude approximately 20 and 10 hours of one-time policy training, respectively. H-SP times exclude a reported one-time warm start of 1,288 seconds. Learn2Clean includes its online training and inference. We report cold-start and reuse costs separately in Appendix B.  Other runtimes follow the end-to-end accounting described for the benchmark. The 16W configuration uses 16 parallel workers.
Reported times are wall-clock seconds per target dataset. Times are end-to-end except for one-time offline components (full report in \supp{B}). CtxPipe and HAIPipe exclude about 20\,h and 10\,h of policy training, respectively, and H-SP excludes a 1{,}288\,s MAB warm start. Learn2Clean has no offline phase, and its time includes online training. The 16W configuration uses 16 parallel workers.

% Unless otherwise marked,

\myparagraph{Other Implementation Details.}
By default, we set the maximum pipeline length $\maxlen = 6$ with $N = 25$ operators, the sampling count $n_\text{sample}=75$, and warm-start the MAB model using $T_\text{mab}=1000$ historical pipeline evaluations from a held-out dataset. Each method is run 5 times with different random seeds (42, 114, 256, 512, 768), and we report the average. All experiments are conducted on a server with two AMD EPYC 9334 32-Core Processors and 256\,GB RAM.

\subsection{RQ1: Effectiveness and Efficiency Analysis}
\label{subsec:effectiveness}
% Table~\ref{tab:main_results} presents the performance comparison in mean accuracy across all 27 datasets. Additional metrics (precision, recall, F1) in Table~\ref{tab:average_result} show consistent improvements.

\myparagraph{Effectiveness.}
% Table~\ref{tab:main_results} reports the per-dataset results and the average performance of H-SP and all baselines across the 27 datasets, measured by four effectiveness metrics: accuracy, precision, recall, and F1. On all four metrics, H-SP significantly outperforms every baseline except ES$^{*}$, attaining the best average accuracy ($0.847$) and the best average rank among search-based, context-aware, LLM-based, and Shapley-based methods. This consistent advantage indicates that modeling data preparation operators as cooperative players and guiding pipeline construction by their position-aware contributions is more effective for DPPC than heuristic search, context-driven, or prompting-based strategies. The only method that surpasses H-SP is ES$^{*}$, which performs random exploration under an extremely large budget of $10{,}000$ evaluations and thus serves as an empirical performance upper bound rather than a practical baseline. Even against this upper bound, the gap of H-SP remains small, while H-SP reaches this level with far fewer evaluations and substantially less runtime. 
Table~\ref{tab:main_results} reports per-dataset accuracy and the averages of four effectiveness metrics. Among practical methods, H-SP attains the highest mean accuracy ($0.847$), precision, recall and F1, and the best average rank. H-SP achieves the best accuracy on 18 of the 27 datasets. A two-sided Wilcoxon signed-rank test with Holm correction rejects equal performance against each of the 14 baselines at $p<0.01$. The Friedman test over the 14 methods that complete all 27 datasets is also
significant ($p<10^{-24}$), the full test report in \supp{E}. The margin over the strongest baselines is moderate: $+1.1$ points over MC-SP and $+2.6$ over LOO. 

\myparagraph{Efficiency.}
The last two rows of Table~\ref{tab:main_results} report the average running time and number of
trained pipelines. Among the methods with accuracy comparable to H-SP, H-SP
runs $17\times$ faster than ES$^*$ and $1.6\times$ faster than MC-SP; L2C and
LOO are faster but clearly less accurate (0.807 and 0.821 vs.\ 0.847). H-SP
trains 2{,}136 pipelines on average, 78.6\% fewer than ES$^*$ and 25\% fewer
than MC-SP. Its two-stage factorization skips whole categories at each
position, and the partial orders and early termination prune the remaining
candidates, keeping the count below the bound in Table~\ref{tab:complexity}.
The reported time excludes the one-time 1{,}288\,s MAB warm start, which is
reusable across datasets that share the operator library and learner.

% H-SP first searches a compact space of category sequences and then refines operators within one category per position, so entire categories are skipped at each position. 

\begin{figure}[t]
    \centering
    \makeatletter
    \if@twocolumn
      \includegraphics[width=\columnwidth]{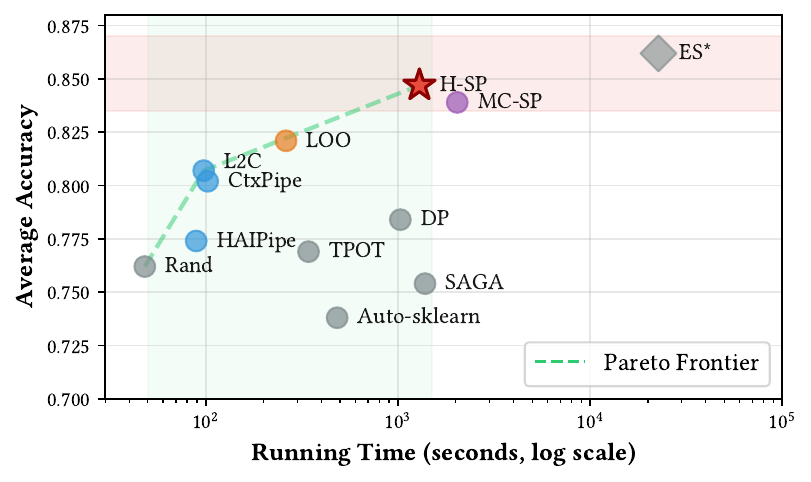}
    \else
      \includegraphics[width=0.6\textwidth]{Figures/efficiency/efficiency_comparison_2.pdf}
    \fi
    \makeatother
    \caption{\textbf{Accuracy--runtime trade-off.} Each marker reports a method's mean accuracy vs. mean runtime over the 27 datasets. H-SP ($\star$) lies on the Pareto frontier: no other method is both faster and more accurate.}
    % \caption{Mean accuracy vs. mean runtime over the 27 datasets.}
    \label{fig:efficiency}
\end{figure}

\begin{figure}[t]
    \centering
    \makeatletter
    \if@twocolumn
      \includegraphics[width=\columnwidth]{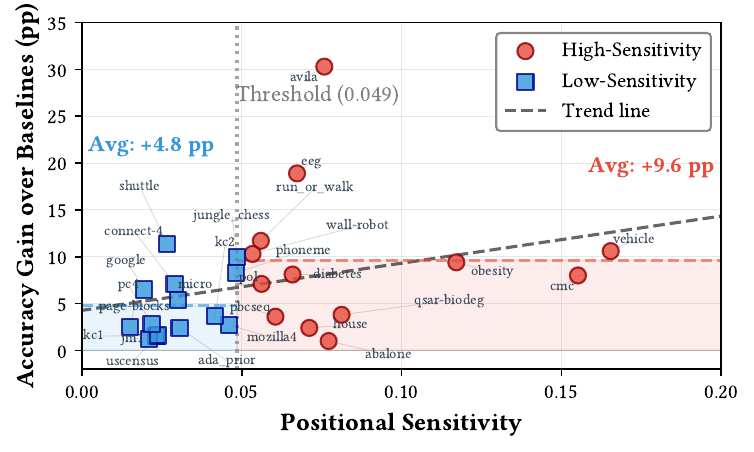}
    \else
      \includegraphics[width=0.6\textwidth]{Figures/pos_sensitivity_vs_acc_gain/pos_sensitivity_vs_mean_acc_gain_pp.pdf}
    \fi
    \makeatother
    % \includegraphics[width=\linewidth]{Figures/pos_sensitivity_vs_acc_gain/pos_sensitivity_vs_mean_acc_gain.pdf}
    % \includegraphics[width=\columnwidth]{Figures/distance_improvement_scatter.pdf}
    % \caption{Relationship between position-variation index and accuracy improvement. Each point represents one dataset. The x-axis shows the position-variation index $d$ (Eq.~\ref{eq:position-distance}), which measures how differently operators contribute across pipeline positions; higher $d$ indicates stronger position-dependent behavior. The y-axis shows the accuracy improvement of H-SP over the position-agnostic LOO baseline. The dashed line marks the median threshold ($d=4.3$). Datasets with high position-variation ($d > 4.3$, shown in blue) achieve an average improvement of $+2.8\%$, compared to only $+0.6\%$ for low-variation datasets (orange). The positive correlation (Spearman $\rho = 0.340$) confirms that position-aware contribution attribution is most beneficial when operator effects are inherently position-sensitive.}
    % \caption{Average accuracy gain of H-SP over the baselines (excluding MC-SP and ES*). For each dataset, we sample 5 pipelines per operator to estimate the positional sensitivity. H-SP achieves higher accuracy gains on datasets with high positional sensitivity.}
    \caption{Accuracy gain of H-SP over the baselines (excluding MC-SP and ES*) vs. positional sensitivity, estimated from 5 sampled pipelines per operator.}   
    \label{fig:distance_improvement}
\end{figure}
%TODO: change caption

% \begin{figure}[h]
%     \centering
%     \includegraphics[width=\columnwidth]{Figures/merged_pvidx_efficiency.pdf}
%     \caption{\err{Accuracy Gain Correlation with Estimation Divergence and Efficiency-Effectiveness Trade-off.}}
%     \label{fig:merged-distance-and-efficiency}
% \end{figure}

% \myparagraph{Effectiveness--efficiency trade-off.}
% Figure~\ref{fig:efficiency} plots model accuracy against total runtime (log scale) for all methods, showing that H-SP sits on the Pareto frontier. Only ES* attains a marginally higher accuracy than H-SP (0.862 vs.\ 0.847), but it does so at a far greater cost, taking more than an order of magnitude longer to run. Conversely, efficient baselines such as L2C and LOO run faster than H-SP, yet their accuracy is substantially lower. The figure omits the four LLM-based baselines. Their running time includes
% remote inference on the model provider's servers, whose latency depends on network conditions and service load, so it is not comparable with the locally measured times of the other methods. Since no competing method is simultaneously faster and more accurate, H-SP offers the most favorable balance whenever pipeline quality is the priority at small-to-moderate scale. The trade-off is also scale-dependent: on datasets with hundreds of thousands of records, per-position estimation becomes expensive, and H-SP needs subsampling and parallelism to stay
% competitive (See scalability analysis in Section~\ref{sec:repetition}).

\myparagraph{Effectiveness--efficiency trade-off.}
Figure~\ref{fig:efficiency} plots mean accuracy against runtime (log scale). H-SP lies on the Pareto frontier: no method is both faster and more accurate, and it is the only method within 1.5 points of ES* at a fraction of its cost. The four LLM-based baselines are omitted because their times include remote inference, whose latency depends on network conditions and service load. The trade-off is scale-dependent: on datasets with hundreds of thousands of records, H-SP needs subsampling and parallelism (Section~\ref{sec:repetition}).

\myparagraph{When does H-SP help most?}
Considering on some datasets, the accuracy gains of H-SP over the baselines might be modest, we aim to determine the types of datasets on which H-SP is likely to significantly outperform the baselines. 
We hypothesize that: \textit{H-SP is more effective on datasets whose performance is more sensitive to changes in operator positions} because H-SP can more effectively capture signals related to operators’ position-aware contributions.

To test this hypothesis, we design a probe called \textit{positional sensitivity}. 
Specifically, the local sensitivity $\Delta(o_i)$ of a dataset to the position of operator $o_i$ is defined as the dispersion of model utility when operator $i$ is inserted to different positions of the pipeline:
\begin{equation*}
\Delta(o_i) = \E_{P\in \libraryplus^{\maxlen - 1}
    | \ordersplus} [CV(\{u(P_{+o_i, j}),\forall j \in [\maxlen]\})]
\label{eq:position_variation}
\end{equation*}
where $CV(\cdot)$ is the Coefficient of Variation metric, and $P_{+o_i, j}$ means the pipeline derived by inserting $o_i$ to the $j$-th position of pipeline $P$. 
A larger local sensitivity $\Delta(o_i)$ indicates that the dataset is more sensitive to changes in the insertion position of $o_i$.
The positional sensitivity of a dataset is defined as the sum of the local sensitivities of all operators.

In Figure~\ref{fig:distance_improvement}, for each dataset, we sample five pipelines to estimate its local sensitivity to each operator $o_i$ and then sum them to obtain the positional sensitivity. 
We observe a clear overall trend: datasets with higher positional sensitivity tend to yield larger average accuracy gains of H-SP over the baselines. To further quantify this relationship, we partition the datasets into a high-sensitivity group and a low-sensitivity group based on the median positional sensitivity. On the high-sensitivity datasets, H-SP achieves an average accuracy gain of 9.6 percentage points over the baselines, compared to only 4.8 points on the low-sensitivity datasets. These results support our hypothesis that H-SP is particularly advantageous on datasets whose performance is more sensitive to operator positions.
In practice, this analysis can also guide DPPC method selection. Specifically, positional sensitivity can be efficiently estimated through pipeline sampling, and the resulting estimate can be used to determine whether H-SP is likely to justify extra cost.

\begin{table}[h]
\centering
\caption{Ablation of H-SP, averaged over the 27 datasets. $\Delta$Acc is the difference in mean accuracy from the full H-SP, in percentage points (pp).}
\label{tab:ablation_unified}
\setlength{\tabcolsep}{3.5pt}
\small
\begin{tabular}{lrrrr}
\toprule
\textbf{Variant} & \textbf{Acc.} & \textbf{$\Delta$Acc (pp)} & \textbf{\#Evals} & \textbf{Time (s)} \\
\midrule
\rowcolor{gray!15}
H-SP (full) & 0.847 & --- & 2,136 & 1,294 \\
\midrule
w/o hierarchy (= MC-SP) & 0.836 & $-$1.1 & 2,845 & 2,036 \\
\midrule
MAB $\rightarrow$ exhaustive selection & 0.827 & $-$2.0 & 14,052 & 7,125 \\
MAB $\rightarrow$ random selection & 0.810 & $-$3.7 & 2,755 & 2,123 \\
Cold-start MAB & 0.839 & $-$0.8 & 2,375 & 1,701 \\
\midrule
w/o early termination ($M{=}\maxlen{=}6$) & 0.837 & $-$1.0 & 3,074 & 2,574 \\
\bottomrule
\end{tabular}
\end{table}

\subsection{RQ2: Ablation Study}
\label{subsec:ablation}

We ablate the three core components of H-SP: hierarchical decomposition, MAB-based representative selection, and early termination, to quantify their individual contributions.

\myparagraph{Without hierarchy.}
We first assess the two-stage hierarchical decomposition by disabling it, which reduces H-SP to MC-SP: PipeShapley values are estimated directly over the entire operator space, without the intermediate category-level abstraction. As reported in Table~\ref{tab:ablation_unified}, removing the hierarchy lowers accuracy to $0.836$ and incurs a $57\%$ slowdown relative to the full H-SP, because the flat operator space inflates the number of utility evaluations. These results confirm that hierarchical decomposition is essential to both the efficiency and the effectiveness of H-SP.

\myparagraph{MAB-related ablation.}
H-SP relies on an MAB Selection strategy to quickly identify a representative operator within each category; we examine its importance by replacing it with three alternatives. \emph{Exhaustive Selection} enumerates all within-category operators and keeps the one with the largest utility improvement; \emph{Random Selection} draws a representative uniformly at random; and \emph{Cold-start MAB} retains the bandit but does not pre-train its parameters on historical data. As shown in Table~\ref{tab:ablation_unified}, Exhaustive Selection is markedly slower yet attains \emph{lower} accuracy than MAB Selection, as greedily maximizing single-step utility tends to get trapped in local optima and is sensitive to the sampling noise of category pipelines. Random Selection matches the efficiency of MAB Selection but suffers a substantial accuracy drop, since every within-category operator is chosen with equal probability regardless of its quality. Cold-start MAB recovers accuracy only to the level of MC-SP, yet remains more efficient than MC-SP owing to the hierarchical structure. Taken together, these comparisons demonstrate that MAB Selection is a key driver of H-SP's superior accuracy--efficiency trade-off.

% \myparagraph{MAB-related ablation.}
% We replace MAB-based representative selection with three alternatives: Exhaustive Selection keeps the within-category operator with the largest utility improvement, Random Selection draws one uniformly at random, and Cold-start MAB keeps the bandit without warm start. Exhaustive Selection is 5.5× slower yet less accurate, because greedy single-step maximization gets trapped in local optima and is sensitive to the sampling noise of category pipelines. Random Selection is as fast as MAB Selection but loses 3.7 points, and Cold-start MAB reaches only MC-SP-level accuracy, though at lower cost. MAB Selection is thus a key driver of H-SP’s accuracy–efficiency trade-off.

% \input{Tables/early-termination-casestudy}

\myparagraph{Without early termination.}
Early termination allows a constructed pipeline to contain fewer than $\maxlen$ operators, so that H-SP adapts the pipeline length to each dataset; disabling it instead forces every pipeline to the maximum length of $\maxlen$ operators. As reported in Table~\ref{tab:ablation_unified}, removing early termination sharply increases runtime while reducing accuracy, because appending unnecessary operators wastes computation and risks over-transformation. With early termination, pipelines contain 3.2 operators on average; on \textsc{obesity}, for example, H-SP stops after two operators, whereas forcing six lowers accuracy from 0.901 to 0.882 (more cases in \supp{E}).

% The case studies in Table~\ref{tab:early_termination_cases} illustrate this effect: H-SP halts after only two operators on \textsc{obesity}, avoiding the large accuracy loss that an overly long fixed-length pipeline would incur. This confirms that PipeShapley-guided early termination is likewise essential to H-SP.

% =============================================================================
% RQ4: SENSITIVITY ANALYSIS
% =============================================================================
\subsection{RQ3: Sensitivity Analysis}
\label{subsec:sensitivity}

\myparagraph{Varying $\maxlen$.}
Figure~\ref{fig:sensitivity_m}(a) and (b) report how the maximum pipeline length $\maxlen$ affects both accuracy and search cost. Accuracy improves as $\maxlen$ increases up to a moderate range and then plateaus or slightly degrades for larger values, whereas the number of evaluations and the runtime keep growing with $\maxlen$. This contrast shows that overly long pipelines over-transform the data without improving its quality, which justifies pairing a reasonably large maximum length with early termination to preserve accuracy at a controlled cost.

% -----------------------------------------------------------------------------
% SAMPLING COUNT
% -----------------------------------------------------------------------------
% \paragraph{Sampling Count $T$ for Shapley Estimation.}
% Since exact Shapley value computation is exponential, we use Monte Carlo sampling. Figure~\ref{fig:sensitivity_t} shows how the sampling count $T$ affects accuracy and stability.

% \begin{figure}[htbp]
%     \centering
%      \includegraphics[width=\columnwidth]{Figures/sensitivity_m.pdf}
%     \caption{Sensitivity to maximum pipeline length $\overline{M}$.}
%     \label{fig:sensitivity_m}
% \end{figure}

% \begin{figure}[htbp]
%     \centering
%      \includegraphics[width=\columnwidth]{Figures/merged_mt.pdf}
%     \caption{Sensitivity to maximum length \overline{M}, Monte Carlo samples $n_{\text{sample}}$ and MAB samples $T_{\text{mab}}$.}
%     \label{fig:sensitivity_m}
% \end{figure}

\begin{figure}[t]
    \centering
    \includegraphics[width=0.6\linewidth]{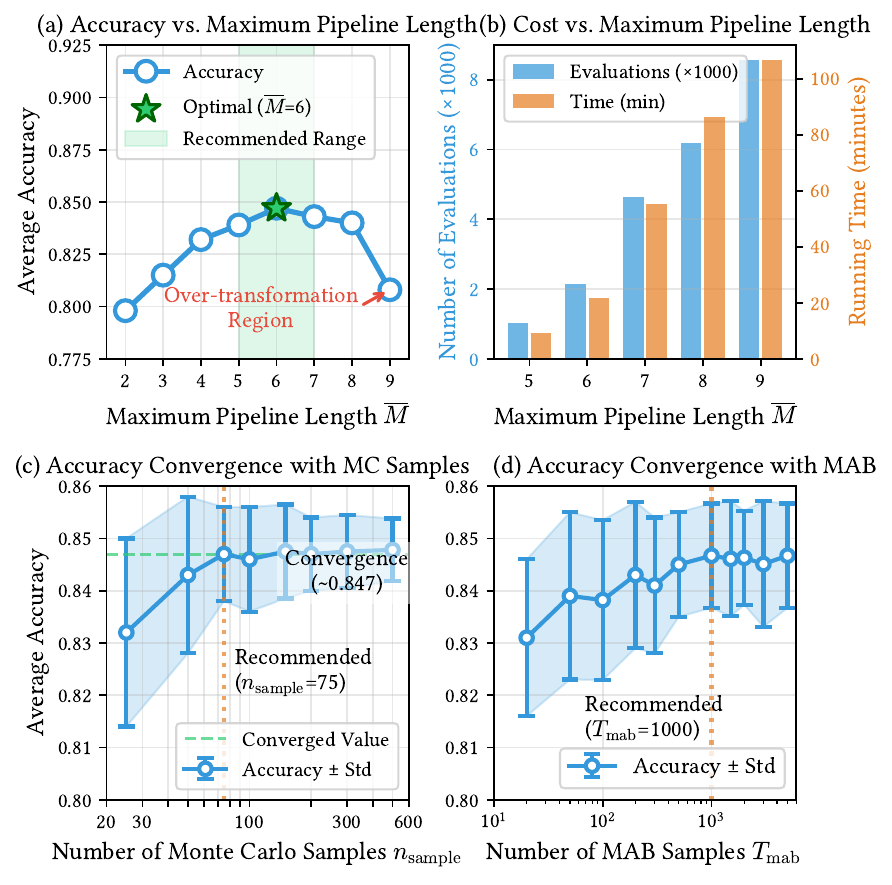}
    \caption{Sensitivity analysis of H-SP. (a) Average accuracy peaks at $\protect\maxlen = 6$ and degrades for larger values. (b) Evaluation count and runtime grow with $\protect\maxlen$. (c) Accuracy converges at $n_{sample} \approx 75$. (d) Accuracy saturates at $T_{mab} \approx 1000$.}
    \label{fig:sensitivity_m}
\end{figure}

\begin{table}[t]
\centering
\caption{\fixed{Model accuracy with different learning algorithms $\mathcal{A}$, averaged over the 27 datasets. Bold marks the best per column. 
% using the
% same 80/20 train/test split as in Table~\ref{tab:main_results}.
% Each pipeline is constructed using Logistic Regression as the
% evaluation model (our default setting); the resulting pipeline is
% then applied to train each alternative learning algorithm.
% Hyperparameters: Random Forest (100 estimators, default);
% XGBoost (100 rounds, max\_depth\,$=$\,6, learning\_rate\,$=$\,0.3);
% MLP (hidden layers\,$=$\,(100), max\_iter\,$=$\,100).
% All other hyperparameters use scikit-learn/xgboost defaults
% with fixed random seed\,$=$\,0.
}
}
\label{tab:learning_algorithms}
\small
\setlength{\tabcolsep}{3.7pt}
\begin{tabular}{l|cccc}
\toprule
\textbf{Setup}       & \textbf{LogisticRegression} & \textbf{RandomForest} & \textbf{XGBoost} & \textbf{MLP} \\ \midrule
Random               & 0.762       & 0.778       & 0.786        & 0.784        \\
TPOT                 & 0.769       & 0.835       & 0.830        & 0.819        \\
AUSK                 & 0.738       & 0.804       & 0.827        & 0.819        \\
L2C                  & 0.807       & 0.846       & 0.856        & 0.834        \\
DiffPrep             & 0.784       & N/A$^*$     & N/A$^*$      & 0.783        \\
SAGA                 & 0.754       & 0.812       & 0.821        & 0.761        \\
CtxPipe              & 0.802       & 0.800       & 0.799        & 0.805        \\
HAIPipe               & 0.774       & 0.780       & 0.823        & 0.774        \\
CatDB                & 0.762       & 0.781       & 0.786        & 0.759        \\
ML-Agent             & 0.733       & 0.827       & 0.816        & 0.798        \\
GPT4o-ZS             & 0.761       & 0.810       & 0.820        & 0.814        \\
GPT4o-FS             & 0.791       & 0.811       & 0.824        & 0.813        \\
LOO                  & 0.821       & 0.845       & 0.849        & 0.824        \\ \midrule
% \textbf{MC-SP}& \textbf{0.836}& \textbf{0.851}& \textbf{0.861}& \textbf{0.839} \\
% \textbf{H-SP} & \textbf{0.847}& \textbf{0.853}& \textbf{0.862}& \textbf{0.839}\\ \bottomrule
MC-SP         & 0.836         & 0.851         & 0.861         & \textbf{0.839} \\
H-SP          & \textbf{0.847}& \textbf{0.853}& \textbf{0.862}& \textbf{0.839}\\ \bottomrule
\end{tabular}
\par\smallskip
\begin{minipage}{\linewidth}
\footnotesize\raggedright * The results for DiffPrep on RandomForest and XGBoost are N/A, as DiffPrep requires differentiable models while RandomForest and XGBoost are non-differentiable.
\end{minipage}
% {\footnotesize
% * N/A since DiffPrep requires differentiable models, while RandomForest and XGBoost are non-differentiable.
% }
% \begin{tablenotes}
%     \footnotesize
%     \item[1] N/A because DiffPrep requires differentiable models, while RandomForest and XGBoost are non-differentiable.
% \end{tablenotes}
% \vspace{1ex}
\end{table}

% \begin{table}[t]
% \centering
% \caption{Impact of different category partitions under the shared
% operator library of $N=25$ operators. Results are averaged over the 27 datasets.}
% \label{tab:category_sensitivity}
% \small
% \setlength{\tabcolsep}{1.5pt}
% \begin{tabular}{p{5.6cm}ccc}
% \toprule
% Categorization method & Acc. & \# Evals & Time (s) \\
% \midrule
% Default categories & 0.847 & 2136 & 1294 \\
% Random (5 categories) & 0.811 & 1663 & 1142 \\
% % Flat (all operators in a single group) & 0.839 & 2845 & 2036 \\
% Split Feature Engineering category into Feature Expansion and Feature Reduction & 0.847 & 1645 & 1205 \\
% \bottomrule
% \end{tabular}
% \end{table}

\begin{table}[t]
\centering
\caption{Impact of the operator taxonomy. Results are averaged over the 27
datasets. The random partition is averaged over 5 independent draws.}
\label{tab:category_sensitivity}
\small
\setlength{\tabcolsep}{1.5pt}
\begin{tabular}{lccrrr}
\toprule
Categorization & $N$ & $K$ & Acc. & \#Evals & Time (s) \\
\midrule
Default functional partition            & 25 & 5 & 0.847 & 2136 & 1294 \\
Random partition                        & 25 & 5 & 0.811 & 1663 & 1142 \\
Split Feat.\ Eng.\ into Expansion / Reduction & 25 & 6 & 0.847 & 1645 & 1205 \\
% Extended library (\,$+$\,Outlier, $+$\,Imbalance\,) & 32 & 7 & 0.XXX & XXXX & XXXX \\
\bottomrule
\end{tabular}
\end{table}

% \noindent\textbf{Sampling Count $n_{\text{sample}}$.} Since exact Shapley value computation 
% is exponential, we use Monte Carlo sampling. Figure~\ref{fig:sensitivity_t}(a) shows 
% that accuracy converges around $n_{\text{sample}} = 75$, balancing estimation 
% quality and computational cost.

% \textbf{Permutation sample count $n_{\text{sample}}$.}
% Figure~\ref{fig:sensitivity_t} indicates that performance stabilizes around $n_{\text{sample}} \approx 75$, offering a good balance between estimation quality and computational cost.

\myparagraph{Varying $\noofsamples$.}
Figure~\ref{fig:sensitivity_m}(c) shows that accuracy improves steadily as the number of permutation samples $n_{\text{sample}}$ increases and becomes stable around $n_{\text{sample}} \approx 75$. Beyond this point, additional samples yield only marginal gains while incurring extra computational cost. We therefore adopt $n_{\text{sample}}=75$ as the default setting, since it offers a favorable balance between estimation accuracy and efficiency.
 
\myparagraph{Varying $T_{mab}$.}
Figure~\ref{fig:sensitivity_m}(d) reports the sensitivity to $T_{\text{mab}}$, the number of utility-improvement samples collected on a held-out dataset to warm-start the MAB model. Performance rises rapidly as $T_{\text{mab}}$ grows from very small values but quickly saturates once it reaches a moderate level. This diminishing-returns pattern indicates that only a small held-out sample is needed to warm-start an effective MAB model, so a large $T_{\text{mab}}$ is not required for the strategy to be effective.
 
\myparagraph{Varying learning algorithms.}
% Table~\ref{tab:learning_algorithms} reports the accuracy of the constructed pipelines when they are paired with learning algorithms other than the one used during search. Although H-SP employs Logistic Regression as the learner during pipeline construction, the resulting pipelines transfer well to non-linear models such as Random Forest, XGBoost, and MLP, showing that H-SP is not tied to a specific downstream model (per-dataset results in \supp{C.5}).
Table~\ref{tab:learning_algorithms} repeats the main experiment with Random Forest, XGBoost, and MLP as the learning algorithms. H-SP attains the highest or tied-highest mean accuracy under all four learners, although its margin over MC-SP narrows to at most 0.2 points for the non-linear learners. This indicates that ShapleyPipe is not tied to a specific downstream model (per-dataset results in Supp. C.5).

% \myparagraph{Varying the category partition.}
% We keep the same $N=25$ operator library used by all compared
% methods and vary only the category partition.
% As shown in Table~\ref{tab:category_sensitivity}, randomly assigning operators to five categories decreases the accuracy of H-SP substantially. In contrast, splitting Feature Engineering into the functionally coherent Feature Expansion and Feature Reduction categories increases the number of categories from five to six while preserving an accuracy of 0.847 and reducing both evaluation count and runtime. These results show that H-SP is not tied to exactly five categories and benefits from functionally coherent category partitions.

\myparagraph{Varying the category partition.} 
Table~\ref{tab:category_sensitivity} compares functional and random partitions of the same 25 operators. Random partitioning reduces mean accuracy from 0.847 to 0.811. Splitting Feature Engineering into expansion and reduction categories retains the reported mean accuracy of 0.847 while reducing evaluations from 2,136 to 1,645 and runtime from 1,294 to 1,205 seconds. The results support the usefulness of functional grouping in this benchmark.

\myparagraph{Varying the extended library.}
% Appendix~D adds seven operators for outlier and class-imbalance handling, producing 32 operators in seven categories. Across the 27 datasets, H-SP obtains a mean accuracy of 0.843, compared with 0.833 for MC-SP, 0.819 for LOO, and 0.692 for Random. These are the methods evaluated in the extended-library comparison. The new operators are selected in runs from 11 datasets. The extended library does not improve the overall mean over the original library, although individual datasets benefit. The experiment demonstrates operation beyond the default taxonomy and identifies variation in the usefulness of the added operators.
We enlarge the library with two
functional families that the CtxPipe library does not cover, Outlier Handling and
Imbalance Handling, giving $N=32$ operators in $K=7$ categories; all the operators and their partial-order constraints are listed in \supp{D}. The
block structure of the coherence matrix persists (\supp{D.3}), and H-SP remains the most accurate practical method at $0.842$, against $0.833$ for MC-SP and 0.819 for LOO (per-dataset results in
\supp{D.4}). Enlarging the library costs H-SP only a re-collected MAB warm start, whereas CtxPipe and HAIPipe require policies to be retrained from scratch. These results show that H-SP is robust to the partition and library size.

\subsection{Other Studies}
\label{sec:repetition}

\begin{table}[t]
\centering
\caption{
Performance on three larger datasets under the protocol
defined in Section~\ref{sec:large-dataset_protocol}. Accuracy is evaluated on the full datasets, and all times are wall-clock seconds. Bold and underline denote the best and runner-up accuracy, respectively. Unless marked otherwise, runtime is end-to-end.
}
\label{tab:large_datasets}
\setlength{\tabcolsep}{3.8pt}
\small
% \resizebox{\mytablewidth}{!}{%
\begin{tabular}{lcccccccc}
\toprule
& \multicolumn{2}{c}{HIGGS}
& \multicolumn{2}{c}{Covertype}
& \multicolumn{2}{c}{Airlines}
& \multicolumn{2}{c}{Average} \\
\cmidrule(lr){2-3}
\cmidrule(lr){4-5}
\cmidrule(lr){6-7}
\cmidrule(lr){8-9}
Method
& Acc. & Time
& Acc. & Time
& Acc. & Time
& Acc. & Time \\
\midrule
Random       & 0.592 & 43   & 0.615 & 203   & 0.598 & 174   & 0.602 & 140  \\
TPOT         & 0.642 & 2365  & 0.714 & 3120  & 0.584 & 2786  & 0.647 & 2757 \\
Auto-sklearn & 0.556 & 2947  & 0.618 & 2995  & 0.572 & 2989  & 0.582 & 2977 \\
SAGA         & 0.651 & 984   & 0.727 & 1357  & 0.584 & 1612  & 0.654 & 1318 \\
DiffPrep     & 0.641 & 1281  & 0.722 & 1902  & 0.632 & 1137  & 0.665 & 1440 \\
Learn2Clean  & 0.622 & 449   & 0.733 & 2275  & 0.638 & 1014  & 0.664 & 1246 \\
CtxPipe$^{\dagger}$
             & 0.643 & 231   & 0.749 & 436   & 0.640 & 265   & 0.677 & 311  \\
HAIPipe$^{\dagger}$
             & \underline{0.670} & 248   & \textbf{0.767} & 427
             & 0.575 & 253   & 0.671 & 309  \\
CatDB        & 0.626 & 125   & 0.732 & 219
             & \textbf{0.663} & 208 & 0.674 & 184 \\
ML-Agent$^\S$\     & N/A   & 378   & N/A   & 452   & N/A   & 437   & N/A   & 422  \\
GPT4o-ZS     & 0.592 & 176   & 0.645 & 293   & 0.637 & 198   & 0.625 & 222  \\
GPT4o-FS     & 0.637 & 182   & 0.711 & 346   & 0.632 & 201   & 0.660 & 243  \\
LOO          & 0.642 & 71    & 0.718 & 388   & 0.638 & 207   & 0.666 & 222  \\
MC-SP        & \textbf{0.706} & 3280
             & 0.758 & 4131
             & 0.641 & 3736
             & \underline{0.702} & 3716 \\
H-SP$^{\ddagger}$
             & \textbf{0.706} & 3593
             & \underline{0.761} & 3825
             & \underline{0.647} & 3533
             & \textbf{0.705} & 3650 \\
H-SP$_{\mathrm{16W}}^{\ddagger}$
             & \textbf{0.706} & 913
             & \underline{0.761} & 747
             & \underline{0.647} & 1012
             & \textbf{0.705} & 891 \\
\bottomrule
\end{tabular}
% }

\vspace{2pt}
\begin{minipage}{\mytablewidth}
\footnotesize
% $^{\dagger}$ CtxPipe and HAI-AI report online inference time only and exclude approximately 20 h and 10 h of one-time offline policy training, respectively. Learn2Clean includes its online training and inference cost. Unless otherwise marked, the remaining baseline times are end-to-end.

% $^{\ddagger}$ The H-SP rows report per-dataset online
% construction-and-evaluation time and exclude the one-time MAB warm start based on $T_{\mathrm{mab}}=1000$ historical pipeline evaluations. The suffix ``16W'' denotes 16 parallel workers.

% $^\S$\ ML-Agent does not produce a valid pipeline on these datasets, the reported time is what it measures, time until failure.
$^{\dagger}$Online inference only (Section~\ref{subsec:setup}). $^{\ddagger}$Excludes the one-time MAB warm start; 16W: 16 parallel workers. $^\S$ML-Agent produces no valid pipeline; the time is measured until failure.
\end{minipage}
\end{table}

\myparagraph{Scalability on larger datasets.} We apply the large-dataset protocol of Section~\ref{subsec:setup} to HIGGS, Covertype, and Airlines (up to 581K records);
Table~\ref{tab:large_datasets} reports accuracy and wall-clock time. Searching on a 10K-record subsample is the main enabler: it cuts H-SP's average search time by $4.7\times$ (from 17,220\,s to 3,650\,s) without changing the average full-data accuracy (0.705; \supp{B.1}), and 16 workers reduce it further to 891\,s. With 16 workers, H-SP is faster on average than SAGA, DiffPrep, Learn2Clean, TPOT, and Auto-sklearn, and more accurate than each of them on every large dataset.

% Two complementary factors keep H-SP
% tractable at this scale. First, searching on a 10K-record subsample reduces H-SP's average online time from 17,220\,s with full-data search to
% 3,650\,s ($4.7\times$), while leaving the average full-data accuracy unchanged at 0.705 (Appendix~B.1). Second, the utility evaluations within each pipeline position are mutually independent; distributing them over 16 workers reduces the average online time further to 891\,s ($4.1\times$; Appendix~B.3). With 16 workers, H-SP is faster on average than SAGA, DiffPrep, Learn2Clean, TPOT, and Auto-sklearn, and more accurate than each of them on every large dataset. 

H-SP does not dominate on every dataset, however. On Airlines, CatDB reaches
higher accuracy (0.663 vs.\ 0.647) in 208\,s, and on Covertype, HAIPipe is
slightly more accurate (0.767 vs.\ 0.761) with 427\,s of online time. H-SP's
average lead comes mainly from HIGGS, where it improves on the best
non-ShapleyPipe method by 3.6 points. Two factors qualify the comparison. First,
the online times of CtxPipe and HAIPipe exclude about 20\,h and 10\,h of policy
training. For a single new dataset, the cold-start cost is about 10.1\,h for
HAIPipe versus 2,179\,s for H-SP with 16 workers (\supp{B.3}). Second, CatDB is
fast and strong on Airlines, but it is not robust across datasets: on the 27
datasets of Table~\ref{tab:main_results}, it reaches a mean accuracy of 0.762 with an
average rank of 10.2.

% , and the RL policies become cheaper per dataset only after
% about 21 target datasets (122 against H-SP with 16 workers)

\myparagraph{Repeated operators.}
We further examine whether allowing repeated operators is useful in practice. Across five runs on 27 datasets, H-SP constructs 135 pipelines, among which 14 contain repeated operators. For these 14 pipelines, we remove the repeated
occurrences while keeping the remaining operator order unchanged and re-evaluate the resulting pipelines. In 10 of the 14 cases, removing the repeated operators decreases accuracy, showing that operator repetition can provide non-redundant
utility in sequential data preparation. All cases are listed in \supp{C.3}.

\section{Related Work}\label{sec:related}
\myparagraph{Data preparation pipeline construction.} 
Existing methods treat data preparation pipeline construction as a combinatorial optimization problem. 
Search-based approaches (Grid, Random~\cite{bergstra2012random}, Bayesian Optimization~\cite{snoek2012practical, thornton2013auto}) are effective for small spaces but suffer from the curse of dimensionality in discrete operator sequences. \fixed{DiffPrep~\cite{li2023diffprep} formulates pipeline search as a differentiable optimization problem using gradient-based relaxation over discrete operator choices.} Evolutionary algorithms (TPOT~\cite{olson2016evaluation}, \fixed{SAGA~\cite{siddiqi2023saga}}) and Reinforcement Learning (RL) methods (\fixed{Learn2Clean~\cite{berti2019learn2clean}}, CtxPipe~\cite{gao2024ctxpipe}, HAIPipe~\cite{chen2023haipipe}) represent the state-of-the-art. LLM-based methods (CatDB~\cite{catdb}, AutoML-Agent~\cite{automlagent}) instead prompt a language model with dataset metadata to propose pipelines or code.
However, none of the existing methods can systematically explain the marginal contribution that each operator makes at a specific position in the pipeline, making principled operator selection and ordering difficult. 
Therefore, in this paper, we propose PipeShapley to enable position-aware contribution attribution for operators and develop ShapleyPipe as an interpretable framework for operator selection and ordering. 

\myparagraph{Shapley values in machine learning.}
The Shapley value~\cite{shapley1953value} is the unique allocation that satisfies efficiency, symmetry, linearity and the null-player property. It underlies
feature attribution~\cite{lundberg2017unified} and data
valuation~\cite{ghorbani2019data,jia2019knnshapley,kwon2022beta,wang2023databanzhaf}, and it is typically approximated by permutation
sampling~\cite{castro2009polynomial}. Shapley interactions~\cite{singhvi2025usingshapley} capture
joint effects among players. All of these treat players as an unordered set in which each player participates at most once. Games in generalized characteristic function form~\cite{nowak1994shapley,sanchez1997values} let the worth of a coalition depend on the order in which its members join, but they still assign one value per player. Multi-choice games~\cite{hsiao1993shapley} allow several levels of participation, but the
levels are not positions in a sequence. PipeShapley differs from these. It
attributes utility to (operator, position) pairs, allows the same operator to occupy several positions, and admits a conditional form that directly drives sequential construction.

% \myparagraph{Game theory in AutoML.} Game theory has emerging applications in ML optimization, including coalition formation for ensembles~\cite{leino2018influence, vorobeychik2018adversarial} and Nash equilibrium for architecture search~\cite{bender2018understanding, liu2019darts}. Unlike prior work using game theory for resource allocation~\cite{ren2021lotteryfl} or coordination~\cite{li2020system}, ShapleyPipe uses cooperative games to quantify marginal contributions in sequential, position-dependent settings, enabling unified optimization and interpretability.

\myparagraph{Cooperative games for model and pipeline design.}
Shapley-based attributions have also guided design decisions such as feature selection~\cite{cohen2007feature} and ensemble
pruning~\cite{rozemberczki2021ensemble}. These works value the members of a fixed set of components. ShapleyPipe instead values (operator, position) pairs, including repeated operators in a sequence that is still being built, and it uses these values to decide what to build next.

\section{Conclusion}\label{sec:conclusion}
This paper studied the DPPC problem from a cooperative game-theoretic perspective. We proposed PipeShapley, a position-aware contribution attribution metric that interprets the contribution magnitude of each operator at different positions within a data preparation pipeline. Building on this, we introduced ShapleyPipe, a framework for DPPC that achieved interpretable operator selection and ordering. To address the computational tractability of PipeShapley, we proposed H-SP, which reduces the number of pipeline evaluations by grouping operators according to their functionality and sharing contribution pattern signals within each category. Although H-SP is less efficient than some lightweight methods, it lies on the accuracy–runtime Pareto frontier at small-to-moderate scale and, on large-scale data, remains competitive when paired with subsampling and parallelism. Future work will explore richer contribution-attribution metrics for DPPC and extend H-SP to richer operator taxonomies while preserving the efficiency gains that category coherence provides.

\section*{AI-Generated Content Acknowledgement}

During the preparation of this work we used Claude in order to improve language. We reviewed and edited the content as needed and take full responsibility for the content of the publication.

\bibliographystyle{unsrtnat}
\bibliography{main-ref-base}  %%% Uncomment this line and comment out the ``thebibliography'' section below to use the external .bib file (using bibtex) .

%%% Uncomment this section and comment out the \bibliography{references} line above to use inline references.
% \begin{thebibliography}{1}

% 	\bibitem{kour2014real}
% 	George Kour and Raid Saabne.
% 	\newblock Real-time segmentation of on-line handwritten arabic script.
% 	\newblock In {\em Frontiers in Handwriting Recognition (ICFHR), 2014 14th
% 			International Conference on}, pages 417--422. IEEE, 2014.

% 	\bibitem{kour2014fast}
% 	George Kour and Raid Saabne.
% 	\newblock Fast classification of handwritten on-line arabic characters.
% 	\newblock In {\em Soft Computing and Pattern Recognition (SoCPaR), 2014 6th
% 			International Conference of}, pages 312--318. IEEE, 2014.

% 	\bibitem{hadash2018estimate}
% 	Guy Hadash, Einat Kermany, Boaz Carmeli, Ofer Lavi, George Kour, and Alon
% 	Jacovi.
% 	\newblock Estimate and replace: A novel approach to integrating deep neural
% 	networks with existing applications.
% 	\newblock {\em arXiv preprint arXiv:1804.09028}, 2018.

% \end{thebibliography}

\end{document}